\documentclass[draftclsnofoot,onecolumn]{IEEEtran}

\usepackage{amsthm,amsmath,amssymb,cite}
\usepackage{eucal}
\usepackage{xifthen}
\usepackage{mathtools}
\usepackage{enumerate}
\usepackage{microtype}
\usepackage{xspace}
\usepackage{bm}
\usepackage{fancyhdr}
\usepackage{lastpage}
\usepackage[center]{caption}
\usepackage{xcolor}
\usepackage{algpseudocode}
\usepackage{algorithm}

\renewcommand{\Bmatrix}[1]{\begin{bmatrix}#1\end{bmatrix}}
\newcommand{\Pmatrix}[1]{\begin{array}{ll}#1\end{array}}

\DeclarePairedDelimiter\abs{\lvert}{\rvert}
\DeclarePairedDelimiter\Abs{\lvert}{\rvert^2}

\newcommand{\E}[1][]{\ifthenelse{\isempty{#1}}{\mathbb{E}}{\mathbb{E}\left(#1\right)}}
\newcommand{\conv}{\mathrm{conv}}

\newcommand{\defeq}{\triangleq}
\newcommand{\eqdef}{\triangleq}

\newtheorem{proposition}{Proposition}
\newtheorem{remark}{Remark}
\newtheorem{definition}{Definition}
\newtheorem{theorem}{Theorem}
\newtheorem{example}{Example}
\newtheorem{lemma}{Lemma}

\DeclareMathAlphabet{\mathcal}{OMS}{cmsy}{m}{n}
\newcommand{\nn}{\nonumber}

\usepackage{graphicx}
\usepackage[font=small]{caption}
\usepackage[font=small]{subcaption}
\usepackage{MnSymbol}

\long\def\comment#1{}

\newfont{\bbb}{msbm10 scaled 700}

\newfont{\bb}{msbm10 scaled 1100}

\newcommand{\RR}{\mbox{\bb R}}

\newcommand{\dv}{{\bf d}}

\newcommand{\rv}{{\bf r}}

\newcommand{\xv}{{\bf x}}
\newcommand{\yv}{{\bf y}}

\newcommand{\Am}{{\bf A}}

\newcommand{\Ac}{{\cal A}}

\newcommand{\Cc}{{\cal C}}
\newcommand{\Dc}{{\cal D}}
\newcommand{\Ec}{{\cal E}}

\newcommand{\Gc}{{\cal G}}

\newcommand{\Kc}{{\cal K}}

\newcommand{\Mc}{{\cal M}}
\newcommand{\Nc}{{\cal N}}

\newcommand{\Pc}{{\cal P}}
\newcommand{\Qc}{{\cal Q}}
\newcommand{\Rc}{{\cal R}}
\newcommand{\Sc}{{\cal S}}
\newcommand{\Tc}{{\cal T}}
\newcommand{\Uc}{{\cal U}}

\newcommand{\Vc}{{\cal V}}

\renewcommand{\arg}{{\hbox{arg}}}

\usepackage{hyperref}
\hypersetup{
    bookmarks=true,         
    unicode=false,          
    pdftoolbar=true,        
    pdfmenubar=true,        
    pdffitwindow=false,     
    pdfstartview={FitH},    
    pdfnewwindow=true,      
    colorlinks=true,       
    linkcolor=red,          
    citecolor=green,        
    filecolor=magenta,      
    urlcolor=cyan           
}

\begin{document}
\title{Optimality of Treating Interference as Noise: \\A Combinatorial Perspective}
\author{\authorblockN{Xinping Yi and Giuseppe Caire}\\ 
\thanks{This work was presented in part at IEEE ISIT 2015, Hong Kong, China \cite{Yi:TIN}, and Asilomar Conference on Signals, Systems and Computers \cite{Yi:ITLinQ}.}
\thanks{X. Yi and G. Caire are with Communication and Information Theory Chair in Department of Electrical Engineering and Computer Science at Technische Universit{\"a}t Berlin, 10587 Berlin, Germany. (email: {\tt \{xinping.yi, caire\}@tu-berlin.de})}
}

\maketitle

\begin{abstract}
For single-antenna Gaussian interference channels, we re-formulate the problem of determining the Generalized Degrees of Freedom (GDoF) region achievable by treating interference as Gaussian noise (TIN) derived in \cite{TIN} from a combinatorial optimization perspective. We show that the TIN power control problem can be cast into an assignment problem, such that the globally optimal power allocation variables can be obtained by well-known polynomial time algorithms (e.g., centralized Hungarian method or distributed Auction algorithm). Furthermore, the expression of the TIN-Achievable GDoF region (TINA region) can be substantially simplified with the aid of maximum weighted matchings. We also provide conditions under which the TINA region is a convex polytope that relax those in \cite{TIN}. For these new conditions, together with a channel connectivity (i.e., interference topology) condition, we show TIN optimality for a new class of interference networks that is not included, nor includes, the class found in \cite{TIN}.

Building on the above insights, we consider the problem of joint link scheduling and power control in wireless networks, which has been widely studied as a basic physical layer mechanism for device-to-device (D2D) communications. Inspired by the relaxed TIN channel strength condition as well as the assignment-based power allocation, we propose a low-complexity GDoF-based distributed link scheduling and power control mechanism (ITLinQ+) that improves upon the ITLinQ scheme proposed in \cite{ITLINQ} and further improves over the heuristic approach known as FlashLinQ. It is demonstrated by simulation that ITLinQ+ provides significant average network throughput gains over both ITLinQ and FlashLinQ, and yet still maintains the same level of implementation complexity. More notably, the energy efficiency of the newly proposed ITLinQ+ is substantially larger than that of ITLinQ and FlashLinQ, which is desirable for D2D networks formed by battery-powered devices.
\end{abstract}


\newpage
\section{Introduction}  \label{sec:intro}

Power control and treating interference as Gaussian noise (TIN) is one of the most well-known, vastly employed, and yet most attractive
interference management techniques,  due to its low complexity, robustness to channel uncertainty, and to the fact that codes for the single-user 
Gaussian channel are well understood and efficiently implemented. 
Interestingly, it has also been shown that in some cases TIN is optimal or approximately optimal. 
For example, we know that TIN achieves the sum-capacity in the noisy regime of the two-user Gaussian interference channel  \cite{noisyIC1,noisyIC2,noisyIC3}.
In the general $K$-user single-antenna Gaussian interference channel,  Geng {\em et. al} \cite{TIN} have shown that, subject to a certain set of conditions on 
the channel strengths, TIN achieves the optimal {\em Generalized Degrees of Freedom} (GDoF) region,  and achieves the capacity region to within a constant gap,  independent of the channel coefficients and the signal-to-noise ratio (SNR). 
The TIN optimality condition found in \cite{TIN} is simply expressed in words as the fact that,  for each user (i.e., intended transmitter-receiver pair)
the desired signal strength level is no less than the sum of maximum strengths of all interfering signals from the transmitter 
to the other (unintended) receivers, and to the receiver from the other (unintended) transmitters, when all signal strengths are expressed in 
log-scale (e.g., in dB). For future reference, we indicate this condition as the ``GNAJ" condition, from the initials of the authors 
of \cite{TIN}. Under the GNAJ condition, the TIN-Achievable GDoF region (briefly referred to as ``TINA region'') is a convex polytope defined 
by the individual GDoF constraints and by the sum GDoF inequalities corresponding to all possible ordered subsets of users. With the aid of a combinatorial tool named {\em potential graphs}, the $K$-user TINA region was characterized in \cite{TIN} by $\sum_{m=2}^K {K \choose m} (m-1)! \approx (K-1)!$ constraints. 
More recently, it has been also shown by Sun and Jafar in \cite{TIN-Para} that, by a series of transformations of linear programs, 
the sum-GDoF characterization can be translated into a minimum weighted matching problem in combinatorial 
optimization.  As such, the sum-GDoF under the GNAJ condition can be characterized as disjoint cycles partition of the 
interference network.

Such remarkable findings have inspired various related works, such as the TIN optimality of general X-channels \cite{TINX}, 
parallel interference networks \cite{TIN-Para}, and compound interference networks \cite{TIN-Com}. 
In general, the TIN problem consists of two subproblems.
Beyond the TINA region characterization, it is also important to find efficient methods to solve the TIN power control problem, that is, finding the (minimum) transmit powers that achieve a certain desired GDoF-tuple in the TINA region. The TIN power control problem has been open for a long time until a recent progress reported by Geng and Jafar in \cite{TIN-Com}, where a simple yet elegant 
polynomial-time centralized iterative algorithm to find the globally optimal power allocation variables is provided. This centralized algorithm relies still on the representation by potential graphs. 

One may wonder if the potential graph representation is the only path to both TINA region characterization and TIN power control problems. Further, due to the distributed nature of interference channels, decentralized power allocation algorithms are more interesting, desirable and yet challenging.
In addition, it is worth noting that the GNAJ condition was only proven to be sufficient. An interesting counter-example in \cite{TIN} showed that there exist partially-connected (in the sense of channel strength levels) interference channels, such that TIN achieves the optimal GDoF region and yet the GNAJ condition is not satisfied. A natural question then arises as to whether there exists a larger class of networks, including partially-connected ones, such that TIN is GDoF-optimal (i.e., TIN with power control still achieves the optimal GDoF region of the channel).
These questions motivate this work.

In this paper, the optimality of TINA is revisited. 
The TIN optimality problem was formulated in \cite{TIN} by first eliminating power allocation variables using the potential theorem \cite{combinopt}, to establish 
the TINA region in terms of GDoF variables only, and then by finding the optimal power allocation variables for a given GDoF-tuple in the TINA region \cite{TIN-Com}. 
In contrast, we re-formulate this problem in a reversed way, from a combinatorial optimization perspective \cite{combinopt}. 
Interestingly, by first casting power allocation into an assignment problem, the globally optimal power allocation variables corresponding to any feasible GDoF tuple 
in the TINA region can be found by solving the equivalent assignment problem in polynomial time, either in a centralized manner (e.g., Hungarian method \cite{Hungarian1955,munkres1957}) 
or in a distributed one (e.g., Auction algorithm \cite{DGS-Auction}). Inspired by the duality between the assignment and the maximum weighted matching problems 
in combinatorial optimization \cite{assignment_problem}, we can express the TINA characterization in terms of a maximum weighted matching problem. 
In doing so, the TINA region is significantly simplified, requiring only $2^K-1$ constraints instead of $\approx (K-1)!$. 
Interestingly, such a representation also offers an interpretation of the disjoint cycle partition in \cite{TIN-Para}. 
By this new formulation, we show that the TINA region is a convex polytope under a novel channel strength condition that relaxes the GNAJ condition in \cite{TIN}.
This new condition requires that the desired signal strength of each user is no less than the {\em maximum difference} between the sum strength 
of any pair of incoming/outgoing interference signals and the strength of the link between such a pair (all in dB scale). 
Furthermore, together with a connectivity condition, we are able to establish the optimality of TINA for a new class of networks. Such conditions are not included nor include the GNAJ condition \cite{TIN}.

Whereas fascinating from a conceptual point of view, how to translate these results into practice is also of great interest and practical importance to system designers. Device-to-Device (D2D) communication is expected to play an important role in future wireless communication systems
(e.g., 5G), including applications such as car-to-car, machine-to-machine, proximity-based services, 
and multi-hop infrastructureless mesh networks. The physical layer of D2D systems is usually modeled as a Gaussian interference channel. Under the practical constraint of treating interference as Gaussian noise for the sake of complexity and robustness,\footnote{From \cite{lapidoth1996nearest} we know that 
this condition is essentially equivalent to imposing the use of minimum distance decoding at each receiver.}
a long-standing problem consists of controlling the power of the D2D links (transmit-receive pairs) in order to maximize the overall
network throughput.\footnote{Consistently with \cite{TseHanly-MAC}, we use the term ``throughput''  to indicate the time-averaged rate over a long sequence of scheduling time slots. In contrast, the instantaneous rate is the rate achieved in a single slot, for a given set of active users, i.e., links  with positive transmit power.}
The usual approach of guaranteeing a target signal-to-interference plus noise ratio (SINR) to each link turns out
to yield an operating point that can be arbitrarily far from optimal. This is because some bottleneck links may impose
too stringent constraints to the overall network. In contrast, much better network throughput can be achieved
by selecting a subset of active links in each slot and allocating positive power only to these selected links \cite{FlashLinQ,BinPower,ITLINQ}. 
By scheduling the subsets of active links over time, it is possible to achieve individual throughputs such that
some {\em network utility function} is maximized. In turn, the shape of the network utility function determines the desired fairness 
criterion (e.g., see \cite{mo2000fair,neely-fnt}).
Link selection and scheduling has become the subject of intensive research. This problem is closely related to power control, 
since link selection corresponds to allocating zero or positive power to the transmitters. For a general D2D network, this problem is non-convex and, as a mater of fact, has a combinatorial nature. For example, a well-known power control method consists of 
replacing $\log(1 + {\rm SINR})$ with $\log({\rm SINR})$ in the user rate expression, and using Geometric Programming (GP) \cite{GP}. 
However, by neglecting the ``$1+$'' inside the ``$\log$'' one has implicitly forced all links to use positive power, since assigning zero power to some links
would drive the GP objective function to $-\infty$.  
Instead, it is known that generally much better solutions can be found by first selecting a ``good'' subset of 
active links, and then allocating (positive) power only to the selected links. 

Various schemes for link selection have been proposed in the literature, e.g., \cite{hajek1988link,borbash2006wireless,LinQ-TDMA,FlashLinQ,ITLINQ} to name a few. For example, a large number of works is based on constructing an interference conflict graph \cite{conflict}, and then selecting maximal independent sets. These ``maximal independent set scheduling'' schemes are flawed by a fundamentally arbitrary choice of the threshold according to which two links are considered to be in conflict. 
Recently, a distributed link scheduling mechanism called FlashLinQ was proposed in \cite{FlashLinQ}, using a more dynamic link selection policy. Compared to those ``maximal independent set scheduling'' schemes, FlashLinQ dynamically takes both signal and interference strength into account. 
In FlashLinQ, links are ranked in priority order and considered one by one. Whether or not a link will be scheduled depends on 
whether this link does not cause/receive too much interference to/from links of higher priority that have already been selected (i.e., declared active). 
More recently, inspired by the GNAJ condition in \cite{TIN}, the authors in \cite{ITLINQ} proposed a new distributed link scheduling mechanism (referred to as ``ITLinQ'') that provides significant sum throughput gains over FlashLinQ and yet maintains the same level of low-complexity. Instead of comparing the ratio of signal to interference strength of the new link with a fixed threshold as in FlashLinQ, ITLinQ compares the interference level caused to/received from 
existing links with {\em an appropriately chosen exponent} of the signal strength of the new link. It was verified by simulation in \cite{ITLINQ} that ITLinQ outperforms 
FlashLinQ with 28\%-110\% gains for a scenario where up to 4096 links can be scheduled.

As a matter of fact, for general channel strength coefficients, the maximal subset of links satisfying the GNAJ condition may not lead 
to the maximal (weighted) sum throughput or sum GDoF. As will be demonstrated later, our relaxed channel strength conditions 
provide a larger convex polytope TINA region. This provides a generally larger subset of links on which power control can be applied, resulting in
higher weighted sum GDoF. 
As a consequence, we are able to design a new distributed link scheduling and power control mechanism (named ``ITLinQ+''), 
further fine-tuning the decision criterion of link selection. 
It is demonstrated by simulation that, without power control, ITLinQ+ gains 5\%-20\% average sum throughput improvement over ITLinQ with 1024 links, at the expense of limited signaling 
overhead. When we also include power control, the average sum throughput is further enhanced, and more notably, 
the energy efficiency of ITLinQ+ is substantially improved (e.g., 50-100 times improvement for a 10-link D2D network). 
In short, ITLinQ+ improves the sum throughput performance and yet requires much less energy consumption, which is desirable for 
battery-powered D2D communications.  Notice that achieving better or equal throughput with less energy consumption is not a 
contradiction here, since the network is operated in an interference limited regime, such that rate is not immediately and obviously related 
to transmit power.

This paper is organized as follows. In the next section, we present the system model of the general $K$-user Gaussian interference channels, followed by a summary of the main existing results of the approximate optimality of treating interference as Gaussian noise. In Section \ref{sec:formulation}, we reformulate the TIN problem from a combinatorial optimization perspective and we obtain a simplified description of the TINA region. By the simplified TINA region, we are able to identify a relaxed channel strength condition under which the general TINA region is a convex polytope. In Section \ref{sec:linq-pc}, we consider the GDoF-based link scheduling and power control problem, offering a framework in this regard. Driven by this framework, the new decentralized link scheduling and power control mechanism named ITLinQ+ is proposed in Section \ref{sec:algorithms} with detailed implementations. Section \ref{sec:simulations} provides numerical results 
and comparisons with ITLinQ and FlashLinQ for some scenarios of D2D networks. 
We conclude the paper in Section \ref{sec:conclusion}.

\underline{\bf Notation}: 
Throughout this paper, we define $\Kc \defeq \{1,2,\dots,K\}$. Let $A$, $\Ac$, and $\Am$ represent a variable, a set, and a matrix, respectively. In addition, $\Ac^c$ is the complementary set of $\Ac$, and $\abs{\Ac}$ is the cardinality of the set $\Ac$. $\Am_{ij}$ presents the $ij$-th entry of the matrix $\Am$, and $\Am_i$ is the $i$-th row of $\Am$. $A_{\Sc} \defeq \{A_i, i \in \Sc\}$, and $\Ac_{\Sc} \defeq\cup_{i \in \Sc} \Ac_i$. Define $\Ac \backslash a \defeq \{x| x \in \Ac, x \neq a\}$ and $\Ac_1 \backslash \Ac_2 \defeq \{x | x \in \Ac_1, x \notin \Ac_2\}$. 
Logarithms are in base 2. With a bit abuse of notation, $k \ne i \ne j$ means $k \ne i$, $i \ne j$ and $k \ne j$.

\section{System Model}

\subsection{Channel Model}

We consider a $K$-user interference channel where both transmitters (Tx) and receivers (Rx) are equipped with a
single antenna each. We shall refer to the $j$-th Tx-Rx pair as the $j$-th {\em user}. 
At Rx-$j$ ($\forall j \in \Kc\defeq \{1,\dots,K\}$), the received signal at the discrete-time instant $t$ is given by 
\begin{align} \label{channel-tilde}
Y_j(t) = \sum_{i=1}^K h_{ij} \tilde{X}_i(t) + Z_j(t) 
\end{align}
where $\tilde{X}_i(t)$ is the transmitted signal from Tx-$i$ with power constraint $\E [\Abs{\tilde{X}_i(t)}] \le P_i$, $h_{ij}$ is the channel coefficient between Tx-$i$ and Rx-$j$,
$Z_j(t) \sim \Cc\Nc(0,1)$ is the (normalized) additive white Gaussian noise at Rx-$j$. 
Following \cite{TIN}, we translate the signal model in (\ref{channel-tilde}) into an equivalent GDoF-friendly form, given by
\begin{align}  \label{channel-friendly}
Y_j(t) = \sum_{i=1}^K \sqrt{P^{\alpha_{ij}}} e^{j \theta_{ij}} {X}_i(t) + Z_j(t) 
\end{align}
where $X_i(t)=\frac{\tilde{X}_i(t)}{\sqrt{P_i}}$ is the normalized transmitted signal with power constraint $\E [\Abs{{X}_i(t)}] \le 1$, $\sqrt{P^{\alpha_{ij}}}$ and $\theta_{ij}$ are magnitude and phase of the channel coefficient between Tx-$i$ and Rx-$j$, respectively, and the exponent $\alpha_{ij}$ is defined as the corresponding channel strength level
\begin{align}
\alpha_{ij} = \frac{\log(\max\{1,\Abs{h_{ij}} P_i\})}{\log P}
\end{align}
where $P > 1$.
Given a transmit power $P^{r_i}$ at Tx-$i$ with $r_i \le 0$, the Signal to Interference plus Noise Ratio (SINR) achieved by TIN at Rx-$j$ is given by 
$
\frac{P^{\alpha_{jj} + r_j}}{1+ \sum_{i: i \ne j} P^{\alpha_{ij}+r_i}}.
$
We assume that the transmitters know channel strength levels perfectly for power control, and the receivers have access to both the magnitude and phase of channel coefficients.

\subsection{Treating Interference as Noise}

We follow standard definitions for encoding/decoding functions and achievable rates. 
The individual achievable GDoF of message $W_k$ is defined as
$
d_k \defeq \lim_{P \to \infty} \frac{R_k}{\log P}
$
where $R_k$ is the achievable rate of user $k$.  The (optimal) GDoF region $\Pc^*$ is the collection of all achievable GDoF-tuples $(d_1,d_2,\dots,d_K)$.
The TIN-Achievable GDoF (TINA) region defined in \cite{TIN} is the set of all $K$-tuples $(d_1,d_2,\dots,d_K)$ with components satisfying
\begin{align} \label{eq:gdof}
d_j \leq \max\left \{ 0, \alpha_{jj} + r_j - \max \{0, \max_{i: i \ne j} (\alpha_{ij} + r_i)\} \right \}.
\end{align}
for some assignment of the power allocation variables $(r_1, r_2, \ldots, r_K) \in \RR^K_{-}$. 
In the following, we denote the TINA by $\Rc^{\rm TINA}$, where the dependence on the specific network defined by $\{\alpha_{ij}: i,j \in \Kc\}$ is clear from the context. 
From \cite{TIN} we also know that the {\em polyhedral} TINA region is obtained by 
removing the positive part operator\footnote{The positive part of $x$ is $\max\{0,x\}$.} from the right-hand side of (\ref{eq:gdof}). 
Using the potential theorem \cite{combinopt}, the authors of \cite{TIN} are able to find a convex polytope form for the polyhedral TINA region 
for any subnetwork formed by a subset $\Sc \subseteq \Kc$ and its associated desired and interfering links. 
We shall denote such polytope by $\Pc^{\rm TINA}_\Sc$. Since removing the positive part in the right-hand side of (\ref{eq:gdof}) restricts 
the GDoF region, then $\Pc^{\rm TINA}_\Sc$ is achievable by switching off all users in $\Sc^c = \Kc \backslash \Sc$ and by using
TIN for the users in $\Sc$.  We also denote by $\Rc^*$ the {\em optimal} GDoF region of the interference network, 
i.e., the region of GDoF-tuples achievable over any possible coding scheme (not restricted to TIN).

The main results in \cite{TIN} are summarized as below.
\begin{theorem} \label{jafar-thm}
$[$ GNAJ \cite{TIN} $]$
Consider a $K$-user single-antenna Gaussian interference channel with channel strengths $\{\alpha_{ij} : i,j \in \Kc\}$.  
\begin{enumerate}
\item For any subnetwork formed by users in $\Sc \subseteq \Kc$, 
$\Pc^{\rm TINA}_{\Sc}$ can be described by\footnote{We use the term {\em ordered subset} to indicate that order matters, but elements are not repeated. For example, $(1,2,3)$ and $(1,3,2)$ are two rising such subsets for $m = 3$, but $(1,2,2)$ is not valid, because it contains repeated elements.}
\begin{align} \label{TINA-Jafar}
& 0 \le d_k  \le \alpha_{kk}, \forall k \in \Sc, \;\;\;  d_{i} = 0, \forall i \in \Sc^c \nonumber \\ 
& \sum_{k = 0}^{m-1} d_{i_k} \le \sum_{k =0}^{m-1} (\alpha_{i_k i_k} - \alpha_{i_{[k-1]_{{\rm mod} m}} i_k}), \nonumber \\
& \forall \;\; \mbox{ordered subsets} \; (i_0,\dots,i_{m-1})  \in \Sc, 
  \forall \;\; m \in  \{2,\dots,\abs{\Sc}\}.
\end{align}
\item The TINA region of the whole network is given by 
\begin{equation} \label{union-tina}
\Rc^{\rm TINA} = \bigcup_{\Sc \subseteq \Kc} \Pc^{\rm TINA}_{\Sc}. 
\end{equation}
\item If $\forall k \in \Kc$,
\begin{align}  \label{GNAJ-cond}
\alpha_{kk} \ge \max_{i:i \ne k} \{\alpha_{ik}\} + \max_{j:j \ne k} \{\alpha_{kj}\}, 
\end{align}
then TIN is GDoF-optimal, i.e., $\Rc^*= \Rc^{\rm TINA} = \Pc^{\rm TINA}_\Kc$ (the whole region is a single convex polytope). 
\end{enumerate}
\end{theorem}

\begin{remark}
It is easy to see that, for $\Sc=\Kc$, there are in total $\sum_{m=2}^K {K \choose m} (m-1)! \approx (K-1)!$ constraints in \eqref{TINA-Jafar}. 
Since the sum-GDoF $\sum_{k=0}^{m-1} d_{i_k}$ does not depend on the order of the indices, 
for each unordered set of indices $\{i_0, \ldots, i_{m-1}\}$ there are $(m-1)!$ inequalities, of which only one is relevant. However, 
finding which one is relevant involves, in general, extensive search, such that finding a general more compact form
that eliminates redundant inequalities is non-trivial. 
\end{remark}

\section{TIN Problem Reformulation from a Combinatorial Perspective}
\label{sec:formulation}
The expression of the TINA region in \eqref{TINA-Jafar} involves a huge number of constraints, some of which are redundant. However, it is unclear which one is necessary and which one is required, such that we do not really know how to analyze this region.
To make progress in this regard, we re-formulate the TIN problem of \cite{TIN,TIN-Com} from a combinatorial optimization perspective.
By casting the power allocation into an assignment problem, we find an alternative form for the TINA region via its dual -- the maximum weighted matching 
problem \cite{assignment_problem}. Some basic definitions of weighted matching are recalled in Appendix \ref{sec:pre}.

\subsection{Casting Power Allocation into Assignment Problems}
In what follows, we consider a feasible GDoF tuple in $\Pc^{\rm TINA}_{\Sc}$ for any user set $\Sc \subseteq \Kc$, where \footnote{Note that we consider $d_i >0, \ \forall~i \in \Sc$. If $d_i=0$, user pair $i$ will be not activated and we simply remove it from $\Sc$ without affecting others.}
\begin{align} \label{eq:ind-gdof}
d_j = \alpha_{jj} + r_j - \max \{0, \max_{i: i \ne j} (\alpha_{ij} + r_i)\}, \, j\in \Sc
\end{align}
given power allocation parameters $\{r_j, j \in \Sc\}$.
By introducing two sets of auxiliary variables, namely, left labels $\{y_{u_j}\}$ and right labels $\{y_{v_j}\}$
\begin{align}
y_{u_j} &= - r_j\\
y_{v_j} &= \max \{0, \max_{i:i \ne j} (\alpha_{ij} + r_i)\},
\end{align}
the individual achievable GDoF can be rewritten as
\begin{align} \label{scummily}
d_j = \alpha_{jj} - (y_{u_j}+y_{v_j}).
\end{align}
Thus, for $\Sc \subseteq \Kc$, the feasibility of a GDoF-tuple can be guaranteed by the minimization 
of the auxiliary variables sum:
\begin{subequations}
\begin{align}
\min_{\{y_{u_j}, y_{v_j}\}} \quad &\sum_{j \in \Sc} ( y_{u_j}+y_{v_j})\\
{\rm s.t.}~\quad & y_{u_j}+y_{v_j} \ge \alpha_{jj} - d_j, \forall j\in \Sc.
\end{align}
\end{subequations}

In general, a given GDoF-tuple in $\Pc^{\rm TINA}_{\Sc}$ may be achieved by different assignments of the power control variables 
$\{r_j : j \in \Kc\}$. The componentwise minimum configuration corresponding to a given 
target GDoF tuple is referred to as the {\em globally optimal power control assignment}. 
In this case, no users can reduce its transmit power while still achieving the same GDoF-tuple.
Using the fact that, for all $i \ne j$,
\begin{align}
y_{u_i} + y_{v_j} &= - r_i + \max \{0, \max_{i':i' \ne j} (\alpha_{i'j} + r_{i'})\}\\
&\ge - r_i +  \max_{i':i' \ne j} (\alpha_{i'j} + r_{i'})\\
& \ge \alpha_{ij},
\end{align}
we have the following theorem that solves the GDoF-based power control problem for a given feasible GDoF-tuple.
\begin{theorem} \label{theorem:power}
For any $(d_j : j \in \Kc) \in \Pc^{\rm TINA}_{\Sc}$, a feasible power allocation assignment $(r_j, j \in \Sc)$ can be found by 
solving the following linear program:
\begin{subequations} \label{eq:assign}
\begin{align} 
(AP): \quad \min_{\{y_{u_j}, y_{v_j}\}} &\sum_{j \in \Sc} (y_{u_j} + y_{v_j})  \label{eq:assign1}\\
{\rm s.t.}~& y_{u_i} + y_{v_j}  \ge \alpha_{ij}, \forall i \ne j  \label{eq:assign2}\\
               &y_{u_j} + y_{v_j}  \ge \alpha_{jj} - d_j, \forall j \in \Sc \label{eq:assign3} \\
               & y_{u_j} \ge 0, \; y_{v_j} \ge 0, \ \forall j \in \Sc  \label{eq:assign4}
\end{align}
\end{subequations}
where $r_j = -y_{u_j}$, $\forall j \in \Sc$. This linear program can be recognized as a dual formulation of an assignment problem \cite{assignment_problem}, so that the unique globally-optimal power allocation can be found in polynomial time (e.g., $O(K^3)$) using e.g., the (centralized) Hungarian method \cite{Hungarian1955,munkres1957} or the (distributed) Auction algorithm \cite{DGS-Auction}.
\end{theorem}

\begin{remark}
The general solutions of (\ref{eq:assign}) may not lead to the globally optimal power allocation. 
{Due to the relation $r_j = -y_{u_j}$ and $y_{u_j} \ge 0, \; y_{v_j} \ge 0$ for all $j$, we conclude that the problem of finding the globally minimal power allocation is equivalent to that of finding the maximum left label equilibrium (cf., $ \max \sum_{j} y_{u_j}$) or the minimum right label equilibrium (cf., $\min \sum_{j} y_{v_j}$) to minimize the sum of the overall labels (i.e., $\min \sum_{j} (y_{u_j} + y_{v_j})$) in \eqref{eq:assign}.
The Hungarian method has various implementations, but most of them are dedicated 
merely to the minimization of the {\em overall} sum of both left and right labels. 
As said, {\em the maximal left label equilibrium} to minimize the overall sum of labels is more relevant in our context. Fortunately, the Kuhn-Munkres algorithm \cite{Hungarian1955,munkres1957}, a variant of the Hungarian method, offers such an equilibrium in solving the assignment problem. A detailed implementation of the Kuhn-Munkres algorithm with some parameters specified to fit our problem is relegated to Appendix \ref{sec:kuhn} (see Algorithm \ref{alg:Hung}).
It is also worth noting that, according to the equality in \eqref{scummily}, the optimal solution to \eqref{eq:assign} is achieved when the equality of \eqref{eq:assign3} holds. This observation is also added to Algorithm \ref{alg:Hung} as the new termination criterion. 
In particular, for an assignment problem with size $K$, the Kuhn-Munkres algorithm requires at most $K$ rounds of iteration to converge 
to the optimal assignment solution. This fact can be also used to check the feasibility of GDoF tuples: if 
Algorithm \ref{alg:Hung} does not converge to the optimal solution for a given GDoF tuple within $K$ iterations, 
then this GDoF tuple is infeasible.

It is also worthwhile to mention that a distributed Auction algorithm, originally due to Demange, Gale, and Sotomayor \cite{DGS-Auction}, 
achieves the minimum right label equilibrium, whose values are element-wise smaller than any other feasible ones, 
leading to the global optimality of power allocation in a decentralized manner. A detailed implementation is presented in Section \ref{sec:auction} (see Algorithm \ref{alg:Auction}).}
\end{remark}

\subsection{TINA Region Representation}
In the following, starting from the power allocation solution of Theorem \ref{theorem:power} and exploiting the duality between 
assignment and maximum weighted matching problems, 
we shall re-formulate the TINA region in a more useful and compact form. 
First, given a GDoF tuple $(d_1,\dots,d_K)$ and channel strength level values $\{\alpha_{ij}, i,j \in \Kc\}$,
we define the following matrix associated with the assignment problem (\ref{eq:assign}):
\begin{align} \label{eq:inputmatrix}
\Am_{ij} = \left\{ \Pmatrix{\alpha_{ij}, & i \ne j \\ \alpha_{jj} -d_j, & i = j} \right. .
\end{align}

By the duality theory in linear programming, we observe that the dual problem of \eqref{eq:assign} is given by
\begin{subequations} \label{eq:mwm}
\begin{align} 
\max  & \sum_{(i,j) \in \Ec} \Am_{ij} x(i,j),    \\
{\rm s.t.} & \sum_{i \in \Uc :  (i,j) \in \Ec} x(i,j) \le 1,  \\
               & \sum_{j \in \Vc :  (i,j) \in \Ec} x(i,j) \le 1,  \\
               & x(i,j) \in [0,1] \label{eq:integral}.  
\end{align}
\end{subequations}
It is known that this linear program has integer-valued 
optimal solutions for bipartite graphs \cite{matching}. 
Since in our case the graph associated to the transmitters and receivers in $\Sc$ and corresponding intended 
and interfering links is bipartite by construction, 
then \eqref{eq:mwm} coincides with a maximum weighted matching problem, 
obtained by replacing \eqref{eq:integral} with $x(i,j) \in \{0,1\}$.

Next, due to the complementary slackness condition \cite{assignment_problem}, an edge $(i,j)$ belongs to the maximum-weight matching, i.e., $x(i,j)=1$, if and only if 
$y_{u_i}+y_{v_j} = \Am_{ij}$.  
For a given $\Sc \subseteq \Kc$, due to (\ref{scummily}) we have that a feasible GDoF-tuple implies 
equality in \eqref{eq:assign3}. Hence, it follows that the set of feasible GDoF-tuples (i.e, the region $\Pc^{\rm TINA}_\Sc$) coincides with the 
set of all $(d_j : j\in \Sc)$ for which the maximum matching solution of (\ref{eq:mwm}) is 
$\{(j,j), j\in \Sc\}$. By the strong duality theorem, the minimum of sum of all left and right labels $\{y_{u_j}, y_{v_j}\}$ in the primal problem \eqref{eq:assign} is equal to the maximum sum weights of all matchings in the dual problem \eqref{eq:mwm}.
This is the key observation that enables us to provide a more compact form for the TINA region. 

We construct a weighted full-connected bipartite graph $\Gc = (\Kc, \Kc, \Kc \times \Kc)$, where the weight $\alpha'_{ij}$ is 
specified as
\begin{align}  \label{alpha-prime}
\alpha'_{ij} = \left\{ \Pmatrix{\alpha_{ij}, & i \neq j \\ 0, & i=j }\right. .
\end{align}
For any $\Sc \subseteq \Kc$, we define the subgraph $\Gc[\Sc] = (\Sc, \Sc, \Sc \times \Sc)$ with weights 
$\{\alpha'_{ij} : i,j \in \Sc\}$. By the observation above, the sum of $\{\Am_{jj} : j\in \Sc\}$ must be 
no less than $w(\Mc)$ for any matching $\Mc$ of $\Gc[\Sc]$. Hence, we can write
\begin{align}
\sum_{j \in \Sc} (\alpha_{jj} - d_j) \ge \max w(\Mc_{\Sc}) = w(\Mc^*_{\Sc}), 
\end{align}
where $\Mc^*_{\Sc}$ is the matching of $\Gc[\Sc]$ with the maximum weight. This yields 
the following result.
\begin{theorem}
\label{theorem:tina-represent}
Consider a $K$-user single-antenna Gaussian interference channel with channel strengths $\{\alpha_{ij} : i,j \in \Kc\}$.  
For any user subset $\Sc \subseteq \Kc$, $\Pc^{\rm TINA}_{\Sc}$ is given by:
\begin{align} \label{ziofa}
\Pc^{\rm TINA}_{\Sc} = \left\{(d_k:k\in \Kc): \Pmatrix{d_k \ge 0, \ \forall k \in \Sc, \quad d_i=0, \forall i \in \Sc^c\\
\sum_{k \in \Sc'} d_k \le \sum_{k \in \Sc'} \alpha_{kk} - w(\Mc^*_{\Sc'}), \ \forall \Sc' \subseteq \Sc} \right\}
\end{align}
where $w(\Mc^*_{\Sc'})=0$ if $\abs{\Sc'}=1$. This simplified representation is equivalent to the expression in \eqref{TINA-Jafar}.
\end{theorem}
\begin{proof}
See Appendix \ref{proof:tina-represent}.
\end{proof}

\begin{remark}
For individual users, i.e., $\abs{\Sc'}=1$, we have individual GDoF constraints, i.e., $d_k \le \alpha_{kk}$.
Using Theorem \ref{theorem:tina-represent} into (\ref{union-tina}), 
we find that we need only $2^K-1$ non-trivial inequalities, one for each non-trivial subset of $\Kc$, to describe 
the $K$-user TINA region $\Pc_\Kc^{\rm TINA}$, which is significantly less than $\approx (K-1)!$ in \cite{TIN}. 
\end{remark}

\begin{remark}
\cite[Theorem 3]{TIN-Para} states that if the GNAJ condition (\ref{GNAJ-cond}) is satisfied, the sum-GDoF is equal to the best cyclic partition bound. 
Explicitly, the best cyclic partition of a user subset $\Sc$ is a partition $\Sc=\{\Sc_1,\dots,\Sc_p\}$
\footnote{Recall that $\{\Sc_1,\dots,\Sc_p\}$ is a partition of $\Sc$ if $\Sc_i \cap \Sc_j = \emptyset$ $\forall i \ne j$ and $\Sc = \bigcup_{i=1}^p \Sc_i$.}
satisfying
\begin{align} \label{eq:bestcycle}
w(\Mc^*_{\Sc}) = \sum_{i=1}^p w(\Mc^*_{\Sc_i}).
\end{align}
\end{remark}

\begin{figure}[ht]
 \centering
\includegraphics[width=0.4\columnwidth]{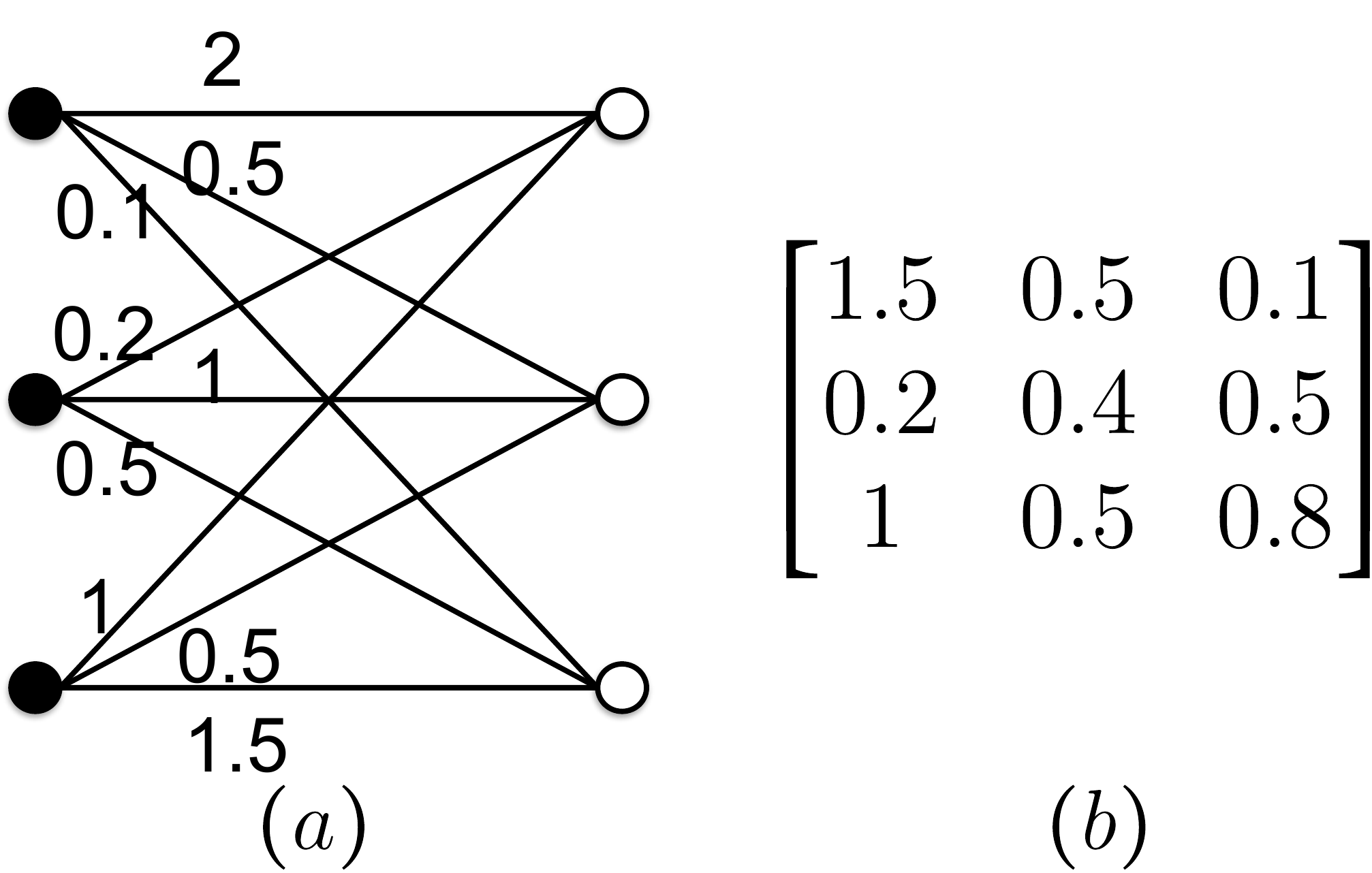}
\caption{(a) A 3-user interference channel, and (b) the input weight matrix of Hungarian method for the GDoF tuple (0.5, 0.6, 0.7).
}
\label{fig:power_ex}
\end{figure}

\begin{example}
\label{ex:power}
\normalfont
We consider the example in \cite[Fig.~8]{TIN-Com} to show the efficiency of our formulation, as shown in Fig.~\ref{fig:power_ex}(a). According to Theorem~\ref{theorem:tina-represent}, the TINA GDoF region is immediately given as
\begin{align*}
\Pc^{\rm TINA}_{\{1,2,3\}} = \{ (d_1,d_2,d_3) :~ &0 \le d_1 \le 2, 0 \le d_2 \le 1, 0 \le d_3 \le 1.5,  \\ &d_1+d_2 \le 2.3, d_2+d_3 \le 1.5 \\ &d_1+d_3 \le 2.4, d_1+d_2+d_3 \le 2.5  \}, 
\end{align*}
which is identical to the expression found in \cite{TIN-Com}. 
In order to solve the power allocation for a given GDoF-tuple (say $(0.5,0.6,0.7)$ in this case), 
we take the weight matrix in Fig.~\ref{fig:power_ex}(b) as the {input of the Kuhn-Munkres algorithm (see Algorithm \ref{alg:Hung} in Appendix \ref{sec:kuhn})} and  we obtain:
\begin{align*}
y_{u_1} = 1.2, \, y_{u_2} = 0.4, \,  y_{u_3} = 0.7, \,  y_{v_1} = 0.3,  \, y_{v_2} = 0, \,  y_{v_3} = 0.1 .
\end{align*}
Thus, the globally optimal power allocation assignment is $r_1=-1.2, r_2=-0.4, r_3=-0.7$, which coincide with what 
found in \cite{TIN-Com}. {The details are relegated to Appendix \ref{sec:kuhn}.}

Clearly, to start Algorithm \ref{alg:Hung}, $y_{u_j}$ and $y_{v_j}$ are initialized respectively with the maximum value of the $j$-th row of $\Am$ and 0. Following the procedure in Algorithm \ref{alg:Hung}, we gradually decrease $y_{u_j}$ and increase $y_{v_j}$ to make sure the constraints in \eqref{eq:assign} satisfied. 
Note that $r_j=-y_{u_j}$ is increasing during this procedure.  
Once we find one solution, it will be the global optimum assignment, because it is impossible to decrease $r_j$ (correspondingly increase $y_{u_j}$) 
and find another solution in the region that we have already explored. 
\hfill $\lozenge$
\end{example}

\subsection{A New TIN Optimality Condition}

Besides the reduction of the number of inequalities, this new formulation enables us to identify a relaxed channel strength condition such that the TINA region is a convex polytope.
\begin{theorem} \label{lemma:sim_region}
Consider a $K$-user single-antenna Gaussian interference channel with channel strengths $\{\alpha_{ij} : i,j \in \Kc\}$.  If
\begin{align} \label{C1}
\alpha_{kk} \ge \max_{i,j:~i,j \neq k} \{\alpha_{ik} + \alpha_{kj} - \alpha'_{ij}\}, \  \forall~k \in \Kc,
\end{align}
where $\alpha'_{ij}$ is defined in \eqref{alpha-prime}, then $\Pc^{\rm TINA}_\Sc$ is monotonically non-decreasing with respect to $\Sc$, i.e., 
if  $\Sc_1 \subseteq \Sc_2 \subseteq \Kc$ then $\Pc^{\rm TINA}_{\Sc_1} \subseteq \Pc^{\rm TINA}_{\Sc_2}$.
Also,  $\Rc^{\rm TINA} = \Pc^{\rm TINA}_{\Kc}$ is a convex polytope. 
\end{theorem}
\begin{proof}
See Appendix \ref{proof:sim_region}.
\end{proof}

\begin{remark}
The newly found channel strength condition is a relaxed version of the GNAJ condition (\ref{GNAJ-cond}), 
because $\alpha'_{ij}$ is non-negative such that if (\ref{GNAJ-cond}) is satisfied, then (\ref{C1})  is satisfied automatically.
When $i=j \ne k$, \eqref{C1} reduces to $\alpha_{kk} \ge \alpha_{ik} + \alpha_{ki}$, $\forall~k,i, {\rm ~s.t.~} k \ne i$. When $i \ne j \ne k$, it reduces to $\alpha_{kk} + \alpha_{ij} \ge \alpha_{ik} + \alpha_{kj}$, for all $k \ne i$, $k \ne j$ and $i \ne j$. As such, in a network where the condition \eqref{C1} is satisfied for all users, for a tuple $(i,j,k)$ with $i \ne j \ne k$, we have either $\alpha_{kk} \ge \alpha_{ik}+\alpha_{kj}$ or $\alpha_{ik}+\alpha_{kj} > \alpha_{kk} \ge \alpha_{ik}+\alpha_{kj} - \alpha_{ij}$. We conclude that, for a user $k$, if $\alpha_{ik}+\alpha_{kj} > \alpha_{kk} \ge \alpha_{ik}+\alpha_{kj} - \alpha_{ij}$, then $\alpha_{kk} \ge \alpha_{ki}+\alpha_{jk}$. It is because, if both $\alpha_{ik}+\alpha_{kj} > \alpha_{kk}$ and $\alpha_{ki}+\alpha_{jk} > \alpha_{kk}$ are satisfied, it leads to $2\alpha_{kk} < \alpha_{ik} + \alpha_{ki}+\alpha_{kj}+\alpha_{jk}$, which conflicts with the fact that $\alpha_{kk} \ge \alpha_{ki}+\alpha_{ik}$ for all $i \ne k$.
\end{remark}

In view of the fact that (\ref{C1}) is a relaxation of (\ref{GNAJ-cond}), the corresponding TINA region 
(although a convex polytope) is not generally optimal. This is because when (\ref{C1}) holds but (\ref{GNAJ-cond}) does not, 
the converse argument to prove GDoF-optimality does not go through.  However, we can exhibit a class of networks 
different from the class identified in \cite{TIN}, for which the TINA region is GDoF-optimal. 
This is a special class of partially connected interference channels satisfying a topological condition given below. 
Interestingly, this class of networks is not included nor includes the class defined by the GNAJ condition. 
The converse proof follows the approach in \cite{TIN} and is presented in Appendix.

\begin{theorem}
\label{theorem:tin}
Consider a $K$-user single-antenna Gaussian interference channel with channel strengths $\{\alpha_{ij} : i,j \in \Kc\}$.  
Assume that (\ref{C1}) holds and, in addition, that for every $\Sc \subseteq \Kc$ with $|\Sc| > 2$, and the corresponding 
fully connected weighted subgraph $\Gc[\Sc] = (\Sc, \Sc, \Sc \times \Sc)$ with weights $\{\alpha'_{ij} : i,j \in \Sc\}$, 
\begin{align} \label{C2}
\exists~ (i,j) \in \Mc^*_{\Sc}, \; \mbox{s.t.} \;  \alpha_{ij}= 0.
\end{align}
Then, $\Rc^* = \Rc^{\rm TINA} = \Pc^{\rm TINA}_\Kc$. 
\end{theorem}
\begin{proof}
See Appendix \ref{proof:tin}.
\end{proof}

\begin{remark}
As the maximum weighted matching may not be unique, Theorem \ref{theorem:tin} holds as long as (\ref{C2}) holds for any one of the maximum matchings.
Condition (\ref{C2}) allows us to establish the optimality  of TINA since, under this condition, we can prove that the converse is tight. 
This, however, is only a sufficient condition and there might be a larger class of networks, including both the subclass defined by 
Theorem \ref{jafar-thm} and the one defined by Theorem \ref{theorem:tin}, for which TIN is GDoF-optimal. 
\end{remark}
\begin{example}
\normalfont
We illustrate the relaxed channel strength condition by the example in Fig.~\ref{fig:tin-ex}. It is easy to verify that the condition (\ref{C1}) holds for the entire network, 
while the original GNAJ condition (\ref{GNAJ-cond}) does not hold for users 1 and 2. 
Note that $\Mc^*=\{(1,3),(2,1),(3,2)\}$ is a (non-unique) maximum weighted matching and contains $\alpha'_{13} = \alpha_{13}=0$, such that also condition (\ref{C2}) holds. 
Thus, from Theorems \ref{lemma:sim_region} and \ref{theorem:tina-represent}, the TINA region of 
this network is the polytope defined by:
\begin{align*}
\Pc^{\rm TINA}_{\{1,2,3\}} = \{ (d_1,d_2,d_3) :~ &0 \le d_i \le 1, \forall i \in \{1,2,3\} \\ &d_1+d_2 \le 1.1, d_2+d_3 \le 1.3 \\ &d_1+d_3 \le 1.2, d_1+d_2+d_3 \le 1.8  \}.
\end{align*}
\hfill $\lozenge$
\begin{figure}[htb]
\centering
\includegraphics[width=0.4\columnwidth]{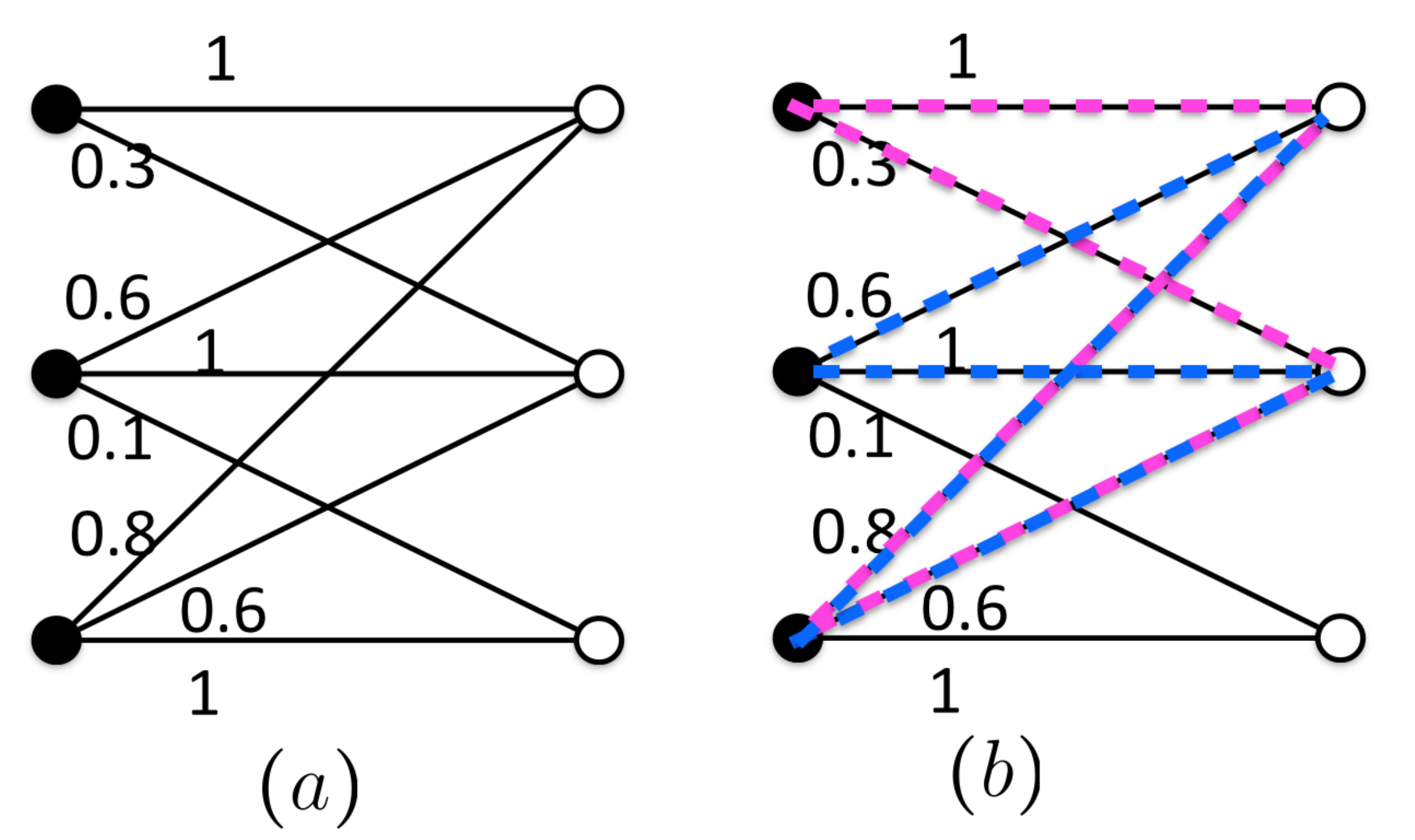}
\caption{ (a) A 3-user IC where TINA is a convex polytope. 
The value associated with each link represents the channel strength level $\alpha_{ij}$ and the missing links correspond to $\alpha_{ij} = 0$.  
(b) The links where the relaxed channel strength condition \eqref{C1} is satisfied while the GNAJ condition does not hold (marked in blue and purple).}
\label{fig:tin-ex}
\end{figure}
\end{example}
\begin{remark}
A subclass of network topologies for which (\ref{C2}) holds is the class of networks that have no perfect matchings  in any unweighted subgraph of $\Gc$ with zero-weight edges removed. A bipartite graph has no perfect matchings if Hall's condition does not  hold \cite{matching}.  The so-called triangular networks in \cite{Yi:TIM} belong to this category.
\end{remark}

\section{A GDoF-based Link Scheduling and Power Control Framework}
\label{sec:linq-pc}

In this section we capitalize on the insight about the TINA region obtained before, in order to develop a framework for link scheduling and 
power control in Gaussian $K$-users interference channels with the constraint that receivers treat interference as (Gaussian) noise. 
As anticipated in Section \ref{sec:intro}, this finds applications in practical interference management of D2D networks, where devices communicate directly to their intended destinations sharing the same channel bandwidth.  In general, the goal is to activate simultaneously a subset of links (i.e., transmitter-receiver pairs) 
with nonzero transmit power, aiming at maximizing some desired system utility function. 
The classical {\em power control} problem (e.g., as formulated in \cite{foschini1993simple,yates1995framework,hanly1995algorithm}) finds the 
componentwise minimum transmit power vector that achieves given target SINRs at the receivers, when such target SINRs are feasible. 
However, this approach does not take into account that in modern TDMA systems the links may not be active in all scheduling slots. 
In contrast, by selecting a subset of links on each slot (scheduling), higher user throughput (i.e., time-averaged rate) can be achieved. 
As anticipated in Section \ref{sec:intro}, a direct application of GP \cite{GP} also does not solve the scheduling problem, since implicitly all links must be 
allocated positive power. Intuitively, these approaches work well when SINRs significantly larger than 1 (0 dB) can be achieved for all the $K$ links. 

A general scheduling framework is provided by considering the user throughputs $T_k = \lim_{t \rightarrow \infty} \frac{1}{t} \sum_{\tau = 1}^t R_k(\tau)$, where 
$R_k(\tau)$ indicates the rate achieved by link $k$ during scheduling slot $\tau$. Let $U(T_1, \ldots, T_K)$ denote a concave componentwise non-decreasing
{\em Network Utility Function} of the user throughputs, and let $\Tc$ denote the achievable throughput region of the system. Then, a general 
{\em Network Utility Maximization} (NUM) problem is given as
\begin{subequations} \label{eq:num}
\begin{align} 
(NUM) : & \quad \max \;\; U(T_1, \ldots, T_K) \label{eq:num1}\\
& \quad {\rm s.t.} \;\;  (T_1, \ldots, T_K) \in \Tc.  \label{eq:num2}
\end{align}
\end{subequations}
In our case, we shall consider a GDoF criterion and replace $T_k$ with $\overline{d}_k = \lim_{P \rightarrow \infty} \frac{T_k}{\log P}$. 
Through an immediate time-sharing argument, we have that the achievable region of throughput-GDoF is the convex hull of $\Rc^{\rm TINA}$, denoted by 
$\conv \Rc^{\rm TINA}$. In general,  $\Rc^{\rm TINA}$ is the union of convex polytopes (see Theorem \ref{jafar-thm}), such that it is not generally convex. 
However, when (\ref{C1}) in Theorem \ref{lemma:sim_region} holds, then $\Rc^{\rm TINA} = \conv \Rc^{\rm TINA} = \Pc^{\rm TINA}_\Kc$. 
Using the GDoF  criterion, the corresponding NUM problem becomes
\begin{subequations} \label{eq:num-gdof}
\begin{align} 
(NUM-GDoF) : & \quad \max \;\; U(\overline{d}_1, \ldots, \overline{d}_K) \label{eq:num1-gdof}\\
& \quad {\rm s.t.} \;\; (\overline{d}_1, \ldots, \overline{d}_K) \in \conv \Rc^{\rm TINA}.  \label{eq:num2-gdof}
\end{align}
\end{subequations}
It turns out that the above problem can be solved by iterating over time (i.e., over the scheduling slot) a sequence of ``instantaneous'' subproblems. 
The following result is quite standard and follows as corollary of the general theory developed for example in 
\cite{neely-fnt,shirani2010mimo,neely2010stochastic,neely2010universal} (and references therein), and shall be stated without proof here for the sake of 
space limitation. 

\begin{theorem} \label{num-solver}
Consider a $K$-user single-antenna Gaussian interference channel with channel strengths $\{\alpha_{ij} : i,j \in \Kc\}$, and 
corresponding TINA region  $ \Rc^{\rm TINA}$. For a sequence of scheduling slots indexed by $t = 1,2,3, \ldots$, consider the following iterative procedure:
\begin{enumerate}
\item Initialize weights $w_k(1) = 1$ for all $k \in \Kc$.
\item For $t = 1, 2, \ldots$, repeat the following two steps: 
\begin{itemize}
\item Compute the GDoF-tuple $(d^*_1(t), \ldots, d^*_K(t))$ solution of the max weighted sum GDoF problem
\begin{subequations} \label{eq:sum-gdof}
\begin{align} 
(SUM-GDoF) : &  \quad \max \;\; \sum_{k\in \Kc} w_k(t) d_k \label{eq:sum1-gdof}\\
& \quad {\rm s.t.} \;\; (d_1, \ldots, d_K) \in \Rc^{\rm TINA}.  \label{eq:sum2-gdof}
\end{align}
\end{subequations}
\item Update the weights according to
\begin{equation} \label{weight-update}
w_k(t+1) = \max\left \{ 0, w_k(t) - d^*_k(t) + a^*_k(t) \right \} , 
\end{equation}
where $(a_1^*(t), \ldots, a^*_k(t))$ is the solution of the convex optimization problem
\begin{subequations} \label{eq:arrival}
\begin{align} 
& \max \;\; V U(a_1, \ldots, a_K) - \sum_{k \in \Kc} w_k(t) a_k \label{eq:arrival1}\\
& {\rm s.t.} \;\; (a_1, \ldots, a_K) \in [0, A_{\max}]^K,  \label{eq:arrival2}
\end{align}
\end{subequations}
where $V > 0$ and $A_{\max} > 0$ are control parameters of the algorithm. 
\end{itemize}
\end{enumerate}
Then, for sufficiently large $A_{\max}$ we have that 
\begin{equation} \label{optimality-NUM}
\lim_{t \rightarrow \infty} U \left ( \frac{1}{t} \sum_{\tau=1}^t d^*_1(\tau), \ldots, \frac{1}{t} \sum_{\tau=1}^t d^*_K(\tau) \right ) \geq U(\overline{d}_1^*, \ldots, \overline{d}^*_K) - \frac{\kappa}{V}, 
\end{equation}
where $(\overline{d}_1^*, \ldots, \overline{d}^*_K)$ is the solution of  the NUM-GDoF problem (\ref{eq:num-gdof}), and
$\kappa$ is a constant that depends on the system parameters but is independent of $V$. Hence, the above iterative scheduling algorithm
can approach the optimal value of (\ref{eq:num-gdof}) by any desired accuracy.
\end{theorem}

\begin{remark} It is also possible to show that the time over which the limit in (\ref{optimality-NUM}) is closely approached
grows as $O(V)$. Therefore, in practice there is a tradeoff between how close we can approach the optimal network utility function value, 
and how quickly the scheduling algorithm converges. Nevertheless, here we are not concerned with this problem, and 
we use Theorem \ref{num-solver} as a general tool to translate a NUM problem in terms of the long-term time averaged rates (or GDoF, in our case)
into a sequence of ``instantaneous'' (i.e., to be solved at each scheduling slot) max weighted sum GDoF problem. 
\end{remark}

It follows that, from now on, we shall be concerned with solving 
the max weighted sum GDoF problem (\ref{eq:sum-gdof}) for an arbitrary set of weights $(w_1, \ldots, w_K)$. 
The ``power control'' aspect of the problem reside in the fact that, when a solution
$(d_1^*(t), \ldots, d_K^*(t))$ of (\ref{eq:sum-gdof}) is found, we must also provide the powers at which
the links have to transmit in order to realize such GDoF point in the TINA region. As anticipated in Section \ref{sec:intro}, 
such transmit powers are generally not unique, and in this case we aim at finding the 
globally optimal power control assignment for the desired GDoF-tuple. 

In what follows, we first introduce the exact GDoF-based solution, and then subsequently simplify it until we could 
obtain an approximation with polynomial-time complexity. In turns, the exact or approximate solver of (\ref{eq:sum-gdof}) can be plugged into 
the iterative scheduling algorithm of Theorem \ref{num-solver} in order to obtain a scheme that works for any suitable 
network utility function. For example, if throughput max-min fairness is desired, we can choose $U(\overline{d}_1, \ldots, \overline{d}_K) = \min_k \overline{d}_k$. Instead,  if proportional fairness is desired, we can choose $U(\overline{d}_1, \ldots, \overline{d}_K) = \sum_{k \in \Kc} \log (\overline{d}_k)$.

\subsection{Exact Joint Solution is Hard}

We re-write (\ref{eq:sum-gdof}) more conveniently in the form:
\begin{subequations} \label{eq:disjunct}
\begin{align}
(DP): \quad \max_{\{d_k\}} \quad & \sum_{k \in \Kc} w_k d_k\\
{\rm s.t.} \quad& (d_1,\dots,d_K) \in \bigcup_{\Sc \subseteq \Kc} \Pc_{\Sc}^{\rm TINA},
\end{align}
\end{subequations}
which can be categorized as an instance of Disjunctive Programming (DP) \cite{disjunctive}. 
The union involves $2^K-1$ nontrivial polyhedra, and $\Pc_{\Sc}^{\rm TINA}$ is described by $2^{\abs{\Sc}}-1$ linear inequalities. As mentioned earlier, the union is not necessarily leading to a convex polytope, and thus the problem is not a convex optimization problem in general. Nevertheless, it can be transformed to an equivalent convex optimization problem by replacing $\bigcup_{\Sc \subseteq \Kc} \Pc_{\Sc}^{\rm TINA}$ with its convex hull
$\Qc = \conv \left( \bigcup_{\Sc \subseteq \Kc} \Pc_{\Sc}^{\rm TINA} \right)$.

The full description of $\Qc$ may require an exponential number of inequalities, yet $\Qc$ has a compact representation in a higher-dimensional space.
The so-called lift-and-project cutting plane method \cite{lift-project,lift-project2} can be employed to offer an exact solution to this problem. The principle consists of three steps: (1) lift the subspace spanned by GDoF tuples into a higher-dimensional space by introducing some auxiliary variables, (2) obtain the compact representation in the form of a set of lift-and-project cutting planes, and (3) project the compact representation onto the original GDoF spanned subspace. These cutting planes are valid for the closure of the convex hull $\Qc$, and can be generated by solving cutting generating linear programs derived from the higher dimensional representation (see \cite{lift-project,lift-project2} and references therein).

Once we obtain the GDoF tuple maximizing the weighted sum-GDoF, the second step is to use either the potential graph based centralized power allocation algorithm found in \cite{TIN} or the assignment problem inspired algorithms (e.g., centralized Hungarian method or distributed Auction algorithm) presented in this paper to find the globally optimal power allocation parameters. This is illustrated in Fig. \ref{fig:jt}.

\begin{figure}[htb]
\centering
\includegraphics[width=0.6\columnwidth]{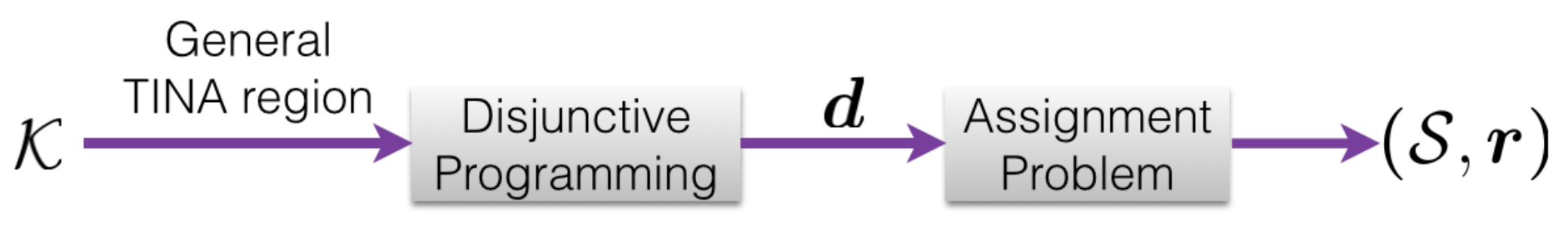}
\caption{The exact solution, where the disjunctive program identifies the optimal GDoF tuple $\dv$ in which the links in the subset $\Sc$ with positive GDoF values are scheduled, and the assignment problem outputs the power allocation vector $\rv$.}
\label{fig:jt}
\end{figure}

Unfortunately, the cutting generating linear programs still involve exponential number of constraints. This fact prohibits the application of this exact solution via disjunctive programming for network of practical size (e.g., a few tens to a few hundreds of D2D links). As such, reasonable approximation and heuristic approaches are desirable, although the global optimality is not guaranteed.

\subsection{Separated Link Scheduling and Power Control }
In view of the complexity of the exact solution, we resort to separated link scheduling and power control. We can first select heuristically a subset of links $\Sc$ whose TINA region $\Pc_\Sc^{\rm TINA}$ contains a GDoF tuple that leads to a reasonably large weighted sum-GDoF for the given user profiles (i.e., weights). Then, we can find this GDoF tuple $\dv$ by optimizing the linear program (i.e., the weighted sum-GDoF) within the TINA region $\Pc_\Sc^{\rm TINA}$. Finally, given this GDoF tuple $\dv$, we can use the assignment inspired algorithms to obtain the globally optimal power allocation parameters $\rv$. This procedure is illustrated in Fig. \ref{fig:sptlp}.

\begin{figure}[htb]
\centering
\includegraphics[width=0.7\columnwidth]{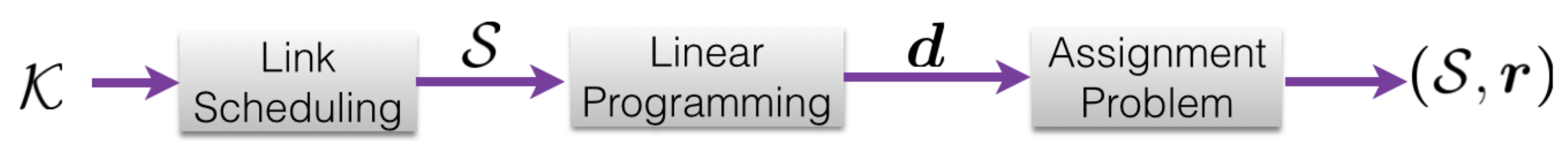}
\caption{The separated solution with a heuristic link scheduling and assignment-inspired power control, where the optimal GDoF tuple is obtained via linear programming over the selected convex polytope.}
\label{fig:sptlp}
\end{figure}

\subsubsection{Link Scheduling}
Recall that in the disjunctive programming \eqref{eq:disjunct}, the objective function can be regarded as a moving hyperplane, and the constraint is the union of polyhedra (convex polytopes). A locally optimal weighted sum-GDoF solution is met when the hyperplane touches one of the vertices as a tangent plane to one convex polytope. The vertices of the largest convex polytope meet such hyperplanes with high probability, such that the weighted sum-GDoF can be maximized with high probability in such a polytope.
As such, a heuristic link scheduling is to select the largest subset of links whose corresponding TINA region is the largest polyhedron (convex polytope) among all polyhedra. 

Similarly to the so-called information theoretic independent sets (ITIS) introduced in \cite{ITLINQ}, we define an independent set according to channel strength levels of our new TIN condition in Theorem \ref{lemma:sim_region}.
\begin{definition}
A user subset $\Sc$ is called an improved information theoretic independent set (referred to hereafter as ``{\rm ITIS+}''), if for any link $k \in \Sc$
\begin{align}
\alpha_{kk} \ge \max_{i,j \in \Sc \backslash \{k\}} \{\alpha_{ik} + \alpha_{kj} - \alpha'_{ij}\}. 
\end{align}
\end{definition}
\begin{remark}
The improvement is in the sense that as long as a subset of users forms an ITIS, it automatically forms an ITIS+. The reverse claim is not true. As such, an ITIS+ contains no less links than an ITIS, and thus the TINA region of the former includes that of the latter.
\end{remark}
The TINA regions of the selected users in ITIS+ are convex polytopes. The largest convex polytope corresponds to the links that form the largest ISIT+. As such, link scheduling turns out to be finding the largest ISIT+.
Although the GDoF optimality does not always hold in the largest ISIT+, it offers the potential to achieve a larger sum throughput than that obtained by ITIS. 

\subsubsection{Power Control}
As said earlier, we can first find the optimal GDoF tuple by maximizing the weighted sum-GDoF for the given set of scheduled links $\Sc$ over the convex polytope $\Pc_{\Sc}^{\rm TINA}$. Once we find such a GDoF tuple, we can apply assignment inspired algorithms to obtain the globally optimal power allocation parameters.

For a given ITIS+ $\Sc$ and the associated user profiles (e.g., weights of individual rates), the GDoF tuple with the maximum weighted sum GDoF can be identified by solving the following linear problem:
\begin{subequations} \label{eq:lp}
\begin{align}
(LP): \quad \max_{d_i} \quad & \sum_{i\in \Sc} w_i d_i\\
{\rm s.t.} \quad& (d_1,\dots,d_K) \in \Pc_{\Sc}^{\rm TINA}
\end{align}
\end{subequations}
For the linear programs, we have polynomial-time complexity algorithms in terms of the numbers of variables and constraints, such as simplex method, interior-point method.
As there is usually an exponential number of constraints (i.e., $2^{\abs{\Sc}}-1$) in $\Pc_{\Sc}^{\rm TINA}$, thus the complexity solving this linear program is still exponential with respect to $\abs{\Sc}$. The problem as to how to reduce the size of constraints by exploiting the special structure of linear program is an interesting problem yet beyond the scope of this paper.

\subsection{Replacing Linear Programming by Geometric Programming}
The exponential-time complexity of the linear problem in \eqref{eq:lp} prohibits its application in large D2D networks. To bypass this challenge, we replace it by a geometric program, which can be solved with polynomial-time complexity, together with a recalculation of GDoF.

As illustrated in Fig.~\ref{fig:sptgp}, given a selected subset of links $\Sc$, Linear Programming (LP) is replaced by Geometric Programming \cite{GP} (GP) to maximize the weighted sum rate at high SNR, and we obtain the corresponding power allocation vector $\rv'$. Basically, this power allocation $\rv'$ achieves the maximum weighted sum throughput yet leads by no mean to the minimum power allocation. To obtain the minimum power allocation, we proceed further and regenerate the GDoF tuple $\dv$ from the power allocation $\rv'$ according to \eqref{eq:ind-gdof}. Finally, assignment-inspired algorithms are utilized to obtain the minimum (globally optimal) power allocation parameters $\rv$. Admittedly, GP can be solely used for power control. Yet, by concatenating with the assignment-inspired algorithms, it lends itself to a lower power consumption, i.e., $\rv' \ge \rv$.

\begin{figure}[htb]
\centering
\includegraphics[width=0.7\columnwidth]{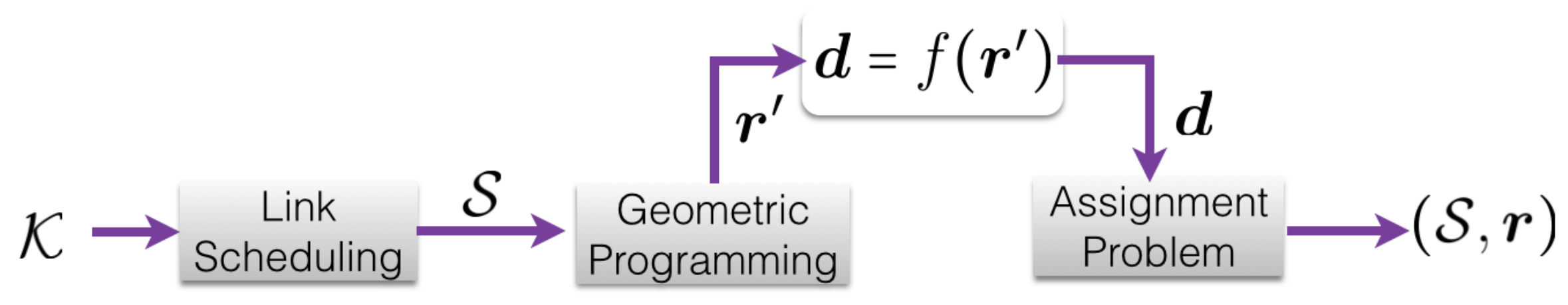}
\caption{The separated solution with a heuristic link scheduling and assignment-inspired power control, where the optimal GDoF tuple is obtained instead via geometric programming.}
\label{fig:sptgp}
\end{figure}

In what follows, we establish the equivalence of such a replacement.
Given the user profiles $\{w_i, i \in \Sc\}$ of a selected use subset $\Sc$, we initially aim at solving the following optimization problem:
\begin{subequations} \label{eq:weighted-rate}
\begin{align}
\max_{\{P_i\}} \quad & \sum_{i \in \Sc} w_i \log (1+{\rm SINR}_i)\\
{\rm s.t.} \quad & {\rm SINR}_i = \frac{G_{ii} P_i}{1+\sum_{j \ne i} G_{ji} P_j}\\
& 0 \le P_i \le 1
\end{align}
\end{subequations}
where $G_{ij}=\Abs{h_{ij}}$ is the channel gain between Tx-$i$ and Rx-$j$.
Let us introduce an auxiliary variable $t_i = \frac{1}{{\rm SINR}_i}$ where $t_i$ is a posynomial function of $\{P_i, i \in \Sc\}$. Thus, the optimization problem at high SNR can be approximated to
\begin{subequations} \label{eq:gp}
\begin{align}
(GP): \quad \min_{\{P_i,t_i\}} \quad& \prod_{i\in {\Sc}} t_i^{w_i}\\
{\rm s.t.}\quad& \frac{1+\sum_{j \ne i} G_{ji} P_j}{G_{ii} P_i} \le t_i, \quad \forall i \in \Sc \\
& 0 \le P_i \le 1
\end{align}
\end{subequations}
which is a geometric program with regard to $\{P_i,t_i,i \in \Sc\}$. A brief description of GP can be found in Appendix.

The equivalence between the LP in \eqref{eq:lp} and the GP in \eqref{eq:gp} is due to the following proposition.
\begin{proposition}
Given a user subset $\Sc$ and the profile $\{w_i, i \in \Sc\}$, if the condition in Theorem \ref{lemma:sim_region} holds, the power allocation $\{P_i, i \in \Sc\}$ by GP in \eqref{eq:gp} is equivalent to LP in \eqref{eq:lp}, in the sense that both approaches achieve the same weighted sum rate at high SNR.
\end{proposition}
\begin{proof}
Let $G_{ij}=\max\{1,P\Abs{h_{ij}} \}=P^{\alpha_{ij}}$ and $P_i = \E [\Abs{{X}_i(t)}] =P^{r_i}  \le 1$. Substituting them into GP, we have
\begin{subequations}
\begin{align}
\min_{\{r_i,t_i\}} \quad& \prod_{i \in {\Sc}} t_i^{w_i}\\
{\rm s.t.}\quad & \frac{1+\sum_{j \ne i} {P^{\alpha_{ji} + r_j}}}{P^{\alpha_{ii}+r_i}} \le t_i \\
& -\infty \le r_i \le 0
\end{align}
\end{subequations}
By replacing $t_i$ with $P^{-d_i}$, we have an equivalent formulation:
\begin{subequations}
\begin{align}
\max_{\{r_i,d_i\}} \quad& \sum_{i \in {\Sc}} w_i d_i\\
{\rm s.t.}\quad & 1+\sum_{j \ne i} P^{\alpha_{ji}+r_j} \le P^{-d_i + \alpha_{ii}+r_i} \\
& r_i \le 0, \quad \forall i \in \Sc
\end{align}
\end{subequations}
Note that the log-sum-exp function can be rewritten as 
\begin{align}
\log (1+\sum_{j \ne i} P^{\alpha_{ji}+r_j}) &= (\max_{j} z_j^i) \log P + \log \sum_{j} P^{z_j^i-\max_{j'} z_{j'}^i} \nn\\
&= (\max_{j} z_j^i + \epsilon_i)  \log P
\end{align}
where $\epsilon_i=\frac{\log \sum_{j} P^{z_j^i-\max_{j'} z_{j'}^i}}{\log P}$ and
\begin{align}
z_j^i=\left\{\Pmatrix{ {\alpha_{ji}+r_j}, & j \ne i \\ 0, & j=i } \right..
\end{align}
It is easily verified that $0 \le \epsilon \le \frac{\log \abs{\Sc}}{\log P}$, and thus the second term in RHS is always bounded within $[0, \log K]$. Thus, we can rewrite the linear program as:
\begin{subequations} \label{eq:gp-equ-form}
\begin{align}
\max_{\{r_i,d_i\}} \quad& \sum_{i \in {\Sc}} w_i d_i\\
{\rm s.t.}\quad & d_i \le \alpha_{ii}+r_i -  \max_{j \ne i} \{0, ({\alpha_{ji}+r_j})\} - \epsilon_i  \\
& d_i \ge 0, \, r_i \le 0, \quad \forall i \in \Sc
\end{align}
\end{subequations}
At high SNR ($P \to \infty$), $\epsilon_i \to 0$. It is not hard to verify that the feasible region of $(d_i: i \in \Sc)$ in the above linear program is exactly the one in \eqref{eq:lp} formulated by taking GDoF metric into account. This completes the proof.
\end{proof}

\begin{remark}
As mentioned in \cite{TIN-Com}, for any rate tuple, there only exists a unique locally optimal power vector, which is also globally optimal (i.e., element-wise minimal), while there are multiple locally optimal power vectors for a GDoF tuple. That being said, multiple locally optimal power vectors lead to the same GDoF tuple, but only one locally optimal power vector is globally optimal, which can be obtained by assignment-inspired algorithms or the algorithms in \cite{TIN-Com}. The concatenation of the geometric program and assignment-inspired algorithms offers maximal sum throughput at high SNR as well as minimal power consumption. Note that they are not contradict, since the network is operated in an interference limited regime such that the sum rate is not immediately and obviously related to transmit power. 
\end{remark}

\section{ITLinQ+: A Decentralized Implementation in D2D Communications}
\label{sec:algorithms}
In D2D communications, smart devices are distributively located such that they have to make their decisions uncoordinatedly. As such, a decentralized implementation of the link scheduling and power control framework in Section \ref{sec:linq-pc} is desirable and of great interest to system designers. 
FlashLinQ \cite{FlashLinQ} and ITLinQ \cite{ITLINQ} are two low-complexity distributed link scheduling mechanisms for D2D networks with reasonable signaling overhead. 

In this section, we propose a decentralized mechanism (referred to as ``ITLinQ+'') for the link scheduling and power control framework aforementioned in Fig. \ref{fig:sptgp}. 
Our proposed ITLinQ+ mechanism consists of three ingredients: (1) a decentralized implementation of link scheduling to find the largest ITIS+, (2) a decentralized GP implementation to find the power allocation vector and the corresponding GDoF tuple, and (3) a distributed Auction algorithm to solve the assignment problem and yield the globally minimal power allocation parameters.
These ingredients are detailed in the following subsections.

\subsection{Decentralized Implementation of Link Scheduling}

To figure out the largest ITIS+, it requires to coordinate all devices and enumerate all possible combinations, which is expensive to coordinate and difficult to enumerate. A decentralized link scheduling criterion with a greedy-flavor independent sets selection, a reasonable amount of signaling, and a comparable complexity to FlashLinQ and ITLinQ, is highly demanded.

The decentralized link scheduling of ITLinQ+ is comprised of two phases: link scheduling and signaling. We first present the link scheduling and then figure out the information that needs to pass.
\subsubsection*{Link Scheduling Phase}
We first order the links according to their priorities (e.g., weights of throughput). To initialize, the links with highest priority will be first selected, and new links will be added to them later. At a certain point, suppose we have selected some links in $\Sc=\{i_1,\dots,i_{k-1}\}$. Whether or not a new link $i_k$ is suitable to be scheduled depends on the following conditions. For the convenience of comparison with FlashLinQ and ITLinQ, we adopt the notations ${\rm SNR}_k \eqdef P^{\alpha_{kk}}$ and ${\rm INR}_{ij} \eqdef P^{\alpha_{ij}}$, where the noise power is normalized.
\begin{itemize}
\item At Tx-${i_k}$, check if the following condition is satisfied:
\begin{align}
{\rm SNR}_{i_k}^\eta \ge  \frac{{\rm INR}_{i_ki_j}} {(\min_{s < k, s \neq j} \{{\rm INR}_{i_si_j}\})^{\gamma}}, \quad \forall~j<k, i_j \in \Sc
\end{align}
where $\eta,\gamma \in [0,1]$ are design parameters, and $\min_{s < k, s \neq j} \{{\rm INR}_{i_si_j}\}$ is the least channel strength level of incoming interfering links of Rx-$i_j$.
\item At Rx-${i_k}$, check if the following condition is satisfied:
\begin{align}
{\rm SNR}_{i_k}^\eta \ge  \frac{{\rm INR}_{i_ji_k}} {(\min_{s < k, s \neq j} \{{\rm INR}_{i_ji_s}\})^{\gamma}}, \quad \forall~j<k, i_j \in \Sc  
\end{align}
where $\min_{s < k, s \neq j} \{{\rm INR}_{i_ji_s}\}$ is the least channel strength level of outgoing interfering links of Tx-$i_j$.
\end{itemize}
If these two conditions are satisfied, then this new link ${i_k}$ can be scheduled, i.e., $\Sc \gets \Sc \cup \{i_k\}$. Note that the minimum value of INR is initialized to be 1 such that the second link $i_2$ is selected if ${\rm SNR}_{i_2}^\eta \ge  \max \{{\rm INR}_{i_1i_2}, {\rm INR}_{i_2i_1} \} $ is satisfied.

\subsubsection*{Link Signaling Phase}
The potential performance improvement of ITLinQ+ over FlashLinQ and ITLinQ is at the expense of additional signaling among devices.
Before the link scheduling phase, we have two rounds of signaling to inform transmitters and receivers the channel strength information, as did in FlashLinQ \cite{FlashLinQ} and ITLinQ \cite{ITLINQ}.
\begin{itemize}
\item In the first round, similarly to FlashLinQ, the transmitted signals from all transmitters are sent with full power $P$ in different frequency bands, such that each receiver is able to estimate the channel strength of the channels to which it is connected. 
\item In the second round, similarly to ITLinQ, pilot signals are sent from all receivers with full power $P$ from each receiver in different frequency bands. As such, the transmitters estimate the channel strength of the channels to which it is connected.
\end{itemize}
At the end of this procedure, for any link $k$, the local channel strength information $\{{\rm INR}_{ki}, \forall~i\}$ and ${\rm SNR}_k$ are accessible at transmitter $k$, and $\{{\rm INR}_{jk},\forall~j\}$ and ${\rm SNR}_k$ at receiver $k$. 

Another additional signaling cost happens at the end of each successful link selection. The transmitter and receiver $i_j$ $(i_j \in \Sc)$ have to inform the next being checked links the minimum interfering channel strength (i.e., ${\min_{i_s \in \Sc, s \neq j} \{{\rm INR}_{i_ji_s}\}}$ and ${\min_{i_s \in \Sc, s \neq j} \{{\rm INR}_{i_si_j}\}}$, respectively). This signaling can be done similarly as the above signaling procedure, or by broadcasting.

To reduce this additional signaling overhead, we can replace both ${\min_{i_s \in \Sc, s \neq j} \{{\rm INR}_{i_si_j}\}}$ and ${\min_{i_s \in \Sc, s \neq j} \{{\rm INR}_{i_ji_s}\}}$ by ${\min_{i_s \in \Sc, s \neq j} \{{\rm INR}_{i_si_j}, {\rm INR}_{i_ji_s}\}}$. The next being checked link $i_k$ is required to compare the minimum value among all cross links in $\Sc$ with the links to which it is connected (i.e., ${\rm INR}_{i_ki_j}$, $\forall i_j \in \Sc$). If the minimum value does not change with the newly selected link, then transmitter or receiver $i_k$ does not need to do anything. Otherwise, link $i_k$ has to inform the next being checked links this updated minimum value.
Finally, the decision criterion of this decentralized link scheduling can be stated in words: A link is scheduled if the interference caused to/received from the already selected higher-priority links is smaller than the product of the signal strength with {\em an exponent $\eta$} of this link and the signal strength with {\em an exponent $\gamma$} of the weakest interfering link among the already selected links. 

\subsection{Decentralized GP Implementation to Find the Optimal GDoF-tuple}
In what follows, we will consider a distributed implementation of GP \cite{GP} where  Rx-$i$ has access to the local knowledge$\{r'_{ji}=\alpha_{ji}+r_j\}_j$ only, i.e., the exponent of received signal power from Tx-$j$. From \eqref{eq:weighted-rate}, we formulate a linear program taking into account the local knowledge 
\begin{subequations}
\begin{align}
\max_{r_i}  \quad & \sum_{i \in \Sc} w_i (\alpha_{ii} + r_i - \max_{j \neq i} \{0, r'_{ji}\} )\\
{\rm s.t.} \quad & r'_{ji} = \alpha_{ji} + r_j, \forall j \in \Sc \backslash \{i\}, \; \forall i\\
 & r_j \le 0,  \forall j \in \Sc 
\end{align}
\end{subequations}

This can be solved by {alternating optimization}. 
First, given $\{\gamma_{ji}\}_j$, we obtain the optimal solution of local variables $r_i$ and $\{r'_{ji}, j \in \Sc \backslash \{i\}\}$. Then, given the updated local variables $r_i$ and $\{r'_{ji}, j \in \Sc \backslash \{i\}\}$, we update $\{\gamma_{ji}\}_j$ again. Keep doing this until it converges.

Introducing Lagrange multipliers only for the coupling constraint, we form a partial Lagrangian
\begin{align}
L=-\sum_{i \in \Sc} w_i (\alpha_{ii} + r_i - \max_{j \neq i} \{0, r'_{ji}\} ) + \sum_{i \in \Sc} \sum_{j \in \Sc \backslash \{i\}} \gamma_{ji} (r'_{ji} - \alpha_{ji} - r_j)
\end{align}
and thus, each user only has to take care of its local partial Lagrangian term, given by
\begin{align}
 L_i (r_i, \{r'_{ji}\}_{j \neq i}; \{\gamma_{ji}\}) &=-w_i (\alpha_{ii} + r_i - \max_{j \neq i} \{0, r'_{ji}\} ) +  \sum_{j \in \Sc \backslash \{i\}} \gamma_{ji} r'_{ji} - \left(\sum_{j \in \Sc \backslash \{i\}} \gamma_{ij} \right) r_i.
\end{align}

The minimization of partial Lagrangian can be done locally by each user in parallel with respect to the primal local variables $r_i$ and $\{r'_{ji}, j \in \Sc \backslash \{i\}\}$ given the knowledge of $\{\gamma_{ji}\}_j$.

The dual variable $\{\gamma_{ji}\}_j$ can be obtained by solving the dual problem
\begin{align}
\max_{\{\gamma_{ji}\}_j}  g(\{\gamma_{ji}\}_j) \defeq \sum_{i}  \min_{r_i, \{r'_{ji}\}_{j \neq i}} L_i (r_i, \{r'_{ji}\}_{j \neq i}; \{\gamma_{ji}\}).
\end{align}
A simple solution of $\{\gamma_{ji}\}_j$ is to update iteratively with the updating rule in $t$ iteration being:
\begin{align}
\gamma_{ji}(t+1) = \gamma_{ji}(t) + \delta(t) (r'_{ji}(t) - r^e_{ji}(t))
\end{align}
where $r^e_{ji}(t)=\alpha_{ji} + r_j(t)$ is the estimation of received signal power exponent, and $\delta(t)$ is a carefully chosen stepsize. Signaling of $\{\gamma_{ji}\}_j$ is needed in each iteration, and the reduction of such signaling overhead can be similarly done as in \cite{GP}.

\subsection{Decentralized Auction Algorithm for Power Allocation}
\label{sec:auction}
The Auction algorithm is an iterative procedure to determine the optimal assignment of a number of products to a number of potential buyers fulfilling their own best interests. It mimicks the sales auction in the business activities in which bids are compared in multiple rounds to make the best offer to the products, with each product going to the highest bidder. The Auction algorithm is an efficient way to solve the assignment problem in a distributed manner. It has many variants, and the algorithm originally proposed by Demange, Gale, and Sotomayor \cite{DGS-Auction} (referred to as ``DGS Auction'') is one of them. Interestingly, DGS Auction algorithm adopts an ascending pricing strategy and converges to the minimum price equilibrium \cite{DGS-Auction}.

In our setting, the transmitters are bidders, and the receivers represent products. Let us look at the assignment problem from an auction perspective. The bidders have access to the local channel strength knowledge, i.e., bidder $i$ only knows $\{\Am_{ij}, j \in \Kc\}$. Here the left label $y_{v_j}$ can be regarded as the price of the product $j$, meaning that a bidder must pay as much as $y_{v_j}$ to obtain the product $j$. For a given price $y_{v_j}$, $\Am_{ij}-y_{v_j}$ can be regarded as the benefit of the bidder $i$ regarding the project $j$. We define profit margin by $y_{u_i} = \max_j \{\Am_{ij}-y_{v_j} \} $. The objective is to determine the best assignment given this local information, such that each bidder is happy to be assigned to a product with the lowest price $y_{v}$ and in turn highest profit margin $y_{u}$. Specifically, it is to minimize the price while maximizing the profit margin, satisfying $y_{u_i} +y_{v_j} \ge \Am_{ij}, \forall i,j$.

An algorithm inspired by DGS Auction is detailed in Algorithm \ref{alg:Auction} where $\epsilon$ is a design parameter.
If $\epsilon$ is small, it requires more rounds of iteration to achieve a reasonably ``almost optimal'' solution. While $\epsilon$ is large, as in real auctions, the bidder may take risks to pay a non-necessarily high price, leading to a suboptimal solution with a faster convergence. 
A demand set $\Dc$ is maintained among bidders to indicate which bidders are not unassigned any product. $O_j$ represents the owner of the product $j$ who successively bids this product.
In each round, the bidder who is not assigned any product will bid his most profitable product, i.e., $j^*= \arg \max_{j} \{\Am_{ij} - y_{v_j} \}$ with $j^*$ being the most profitable product of the bidder $i$. If the associated profit $\Am_{ij^*} - y_{v_{j^*}}$ is negative or $j^*$ is already assigned to the bidder $i$, then we skip this bidder and consider the next unassigned bidder. Otherwise, bidding process starts. If the product $j^*$ was already assigned to another bidder, then this bidder will be added to the demand set and reconsidered later. If this product $j^*$ is free, then it will be assigned to the bidder $i$, and at the same time the price of the product $j^*$ will be raised by $\epsilon$. Keep doing this until every bidder has his product without competitors. The final assignment of bidders and products is the optimal solution to the assignment problem. 

\begin{algorithm}
\caption{A Decentralized Power Allocation Algorithm via Auction Algorithm}
\label{alg:Auction}
\begin{algorithmic}[1]
\Require The bidder $i$ only has the local knowledge $\Am_i = \Bmatrix{\alpha_{i1} & \dots & \alpha_{ii}-d_i & \dots & \alpha_{iK}}$.
\State Initialization: Set $y_{v_j}=0$, $O_j = 0, \forall j$, and $\Dc=\{1,2,\dots,K\}$.
\While {$\Dc \neq \emptyset$} 
\State Choose a bidder $i$ from the demand set, i.e., $\Dc \gets \Dc \backslash \{i\}$
\State For bidder $i$, find the best values in $\{\Am_{ij} - y_{v_j}, \forall j\}$
$$ w_{i} = \max_{j} \{\Am_{ij} - y_{v_j} \}, \quad j^*= \arg \max_{j} \{\Am_{ij} - y_{v_j} \} $$
\If {$w_i \ge \epsilon \; \& \& \; O_{j^*} \ne i$}
\If {$O_{j^*} \neq 0$}
\State Add $i$ into the demand set, i.e., $\Dc \gets \Dc \cup \{O_{j^*}\}$
\EndIf
\State Assign the bidder $i$ with the product $j^*$, i.e., $O_{j^*}=i$
\State Product $j^*$ raises the price by $\epsilon$, i.e., $y_{v_{j^*}} \gets y_{v_{j^*}} + \epsilon$
\EndIf
\EndWhile
\end{algorithmic}
\end{algorithm}

For any set of feasible power allocation parameters $\{r_j=-y_{u_j}\}_j$, there exists an equilibrium price vector $\yv_v=\{y_{v_j}\}_j$. The minimum price equilibrium always exists as long as the corresponding GDoF tuple is feasible. It has been proved in \cite{DGS-Auction} that this DGS Auction algorithm leads to the minimum price equilibrium $\yv_v^*$, meaning that any other feasible price vector $\yv_v$ is element-wise larger than this minimum equilibrium, i.e., $\yv_v^* \le \yv_v$. As such, the globally minimal power allocation parameters can be obtained from $\{r_{i^*}= y_{v_{j^*}} - \Am_{i^* j^*}\}$ where $(i^*,j^*)$ belongs to the optimal bidder-product assignment.

If we make $\epsilon$ sufficiently small such that the price is raised carefully in each round of auction, $\{y_{v_j}\}_j$ obtained by Algorithm \ref{alg:Auction} can achieve arbitrarily close to the minimum equilibrium price.
The algorithm converges within a finite number of iterations. An adaptive price increasing strategy \cite{DGS-Auction1} can be applied to speed up the convergence.

\section{Numerical Results}
\label{sec:simulations}
To demonstrate the gains of our ITLinQ+ mechanism over FlashLinQ and ITLinQ with regard to the sum throughput and energy efficiency, we perform numerical analysis under a similar network setup as in \cite{FlashLinQ,ITLINQ}. 

As the first setup, we only consider the link scheduling without power control to show the benefit of our new decision criterion. We consider a 1 km $\times$ 1 km square area and randomly drop $n$ transmitter-receiver pairs. As in \cite{FlashLinQ,ITLINQ}, the simulated channel follows the LoS model in ITU-1411, and the system parameters are listed in Table \ref{tab:param}. Two scenarios are considered: (1) Scenario 1: The distance of any two paired transmitter and receiver is uniformly distributed in $[5,30]$m, all links are supposed to operate over a 5MHz spectrum, and the maximum transmit power is 20 dBm; (2) Scenario 2: The range of distance is enlarged to $[10,60]$m with a larger maximum transmit power 30 dBm and a wider spectrum bandwidth of 10MHz. We compare ITLinQ+ with no scheduling case where all links are activated, FlashLinQ, as well as ITLinQ with properly chosen parameters $\eta=0.7$ and $M=25$ dB as in \cite{ITLINQ}. In both scenarios, we use the same parameter $\eta=0.9$ and $\gamma=0.1$ for our ITLinQ+.

\begin{table}
\center
\caption{System Parameters}
  \begin{tabular}{ c | c }
    \hline \hline
    Parameters & Values\\ \hline
    Cell range & 1km $\times$ 1km \\ \hline
    Carrier Frequency & 2.4GHz \\ \hline
    Bandwidth & 5MHz and 10 MHz \\ \hline
    Distance (uniformly distributed) & $[5,30]$ and $[10,60]$ m \\ \hline
    Transmit Power & 20 and 30 dBm \\ \hline
    Noise power spectral density & -174 dBm/Hz \\ \hline
    Antenna Height & 1.5 m \\ \hline
    Antenna Gain per Device & -2.5 dB \\ \hline
    Noise Figure & 7 dB \\ \hline \hline
  \end{tabular}
  \label{tab:param}
  \end{table}

Fig. \ref{fig:throughput} shows the sum throughput (bit/sec/Hz) versus the total number of links in the D2D networks ranging from 8 to 1024.
It demonstrates  the significant improvement of ITLinQ+ over FlashLinQ and ITLinQ. For instance, with in total 1024 links, ITLinQ+ achieves 40\% gain in Scenario 1, 60\% gain in Scenario 2 over FlashLinQ, and 20\% gain in Scenario 1 and 40\% gain in Scenario 2 over ITLinQ. \footnote{The gap between ITLinQ and ITLinQ+ might be reduced if one can further fine-tune the parameter $M$ in ITLinQ. Unfortunately, there was no guideline on how to choose this parameter in \cite{ITLINQ}.} As a side remark, the design parameters $\eta, \gamma$ of our ITLinQ+ keep unchanged for various scenarios, which meets the demand of robustness in D2D networks.

\begin{figure}[htb]
 \centering
\includegraphics[width=0.49\columnwidth]{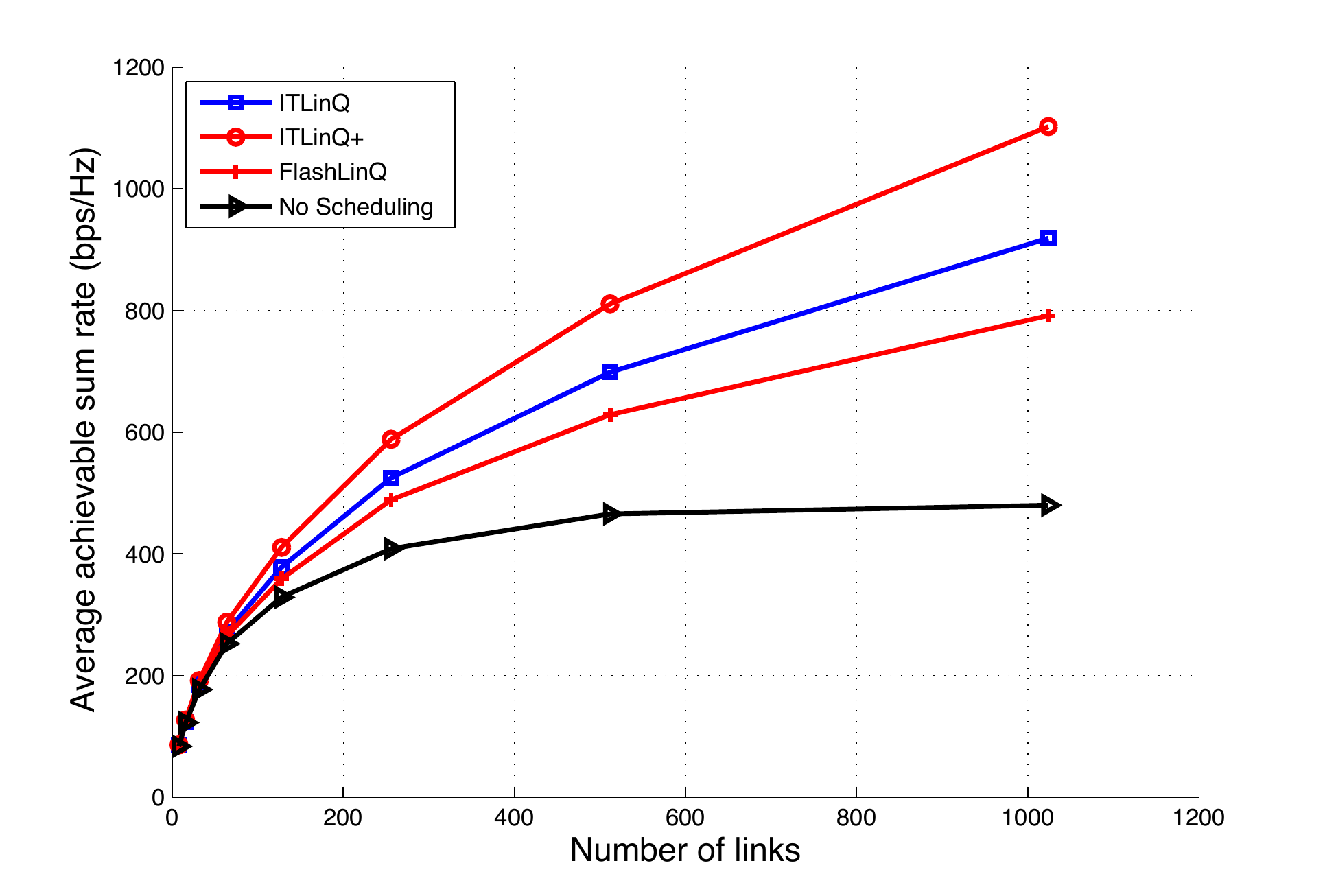}
\includegraphics[width=0.49\columnwidth]{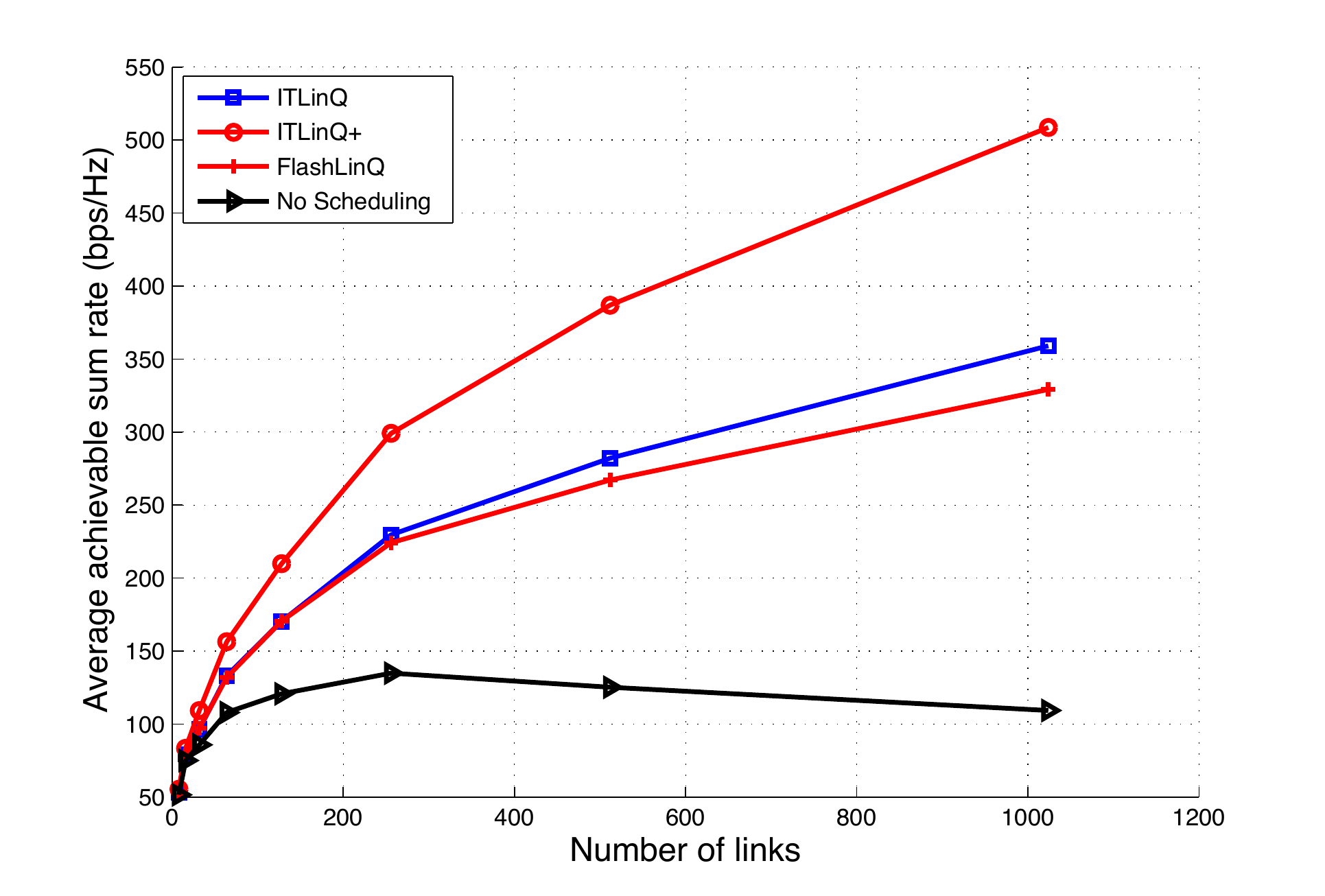}
\caption{Sum throughput comparison among no scheduling, FlashLinQ, ITLinQ and our new ITLinQ+ without power control. The left is for Scenario 1 and the right for Scenario 2.}
\label{fig:throughput}
\end{figure}

To understand better the dramatic energy efficiency improvement of our power control algorithms, we consider a smaller random network with 10 link pairs where the desired and cross link strength levels are uniformly distributed in [1,2] and [0,1], respectively. We compare in Fig. \ref{fig:vssnr} the sum throughput and energy efficiency (bits per Joule) versus SNR.
\begin{figure}[htb]
 \centering
\includegraphics[width=0.45\columnwidth]{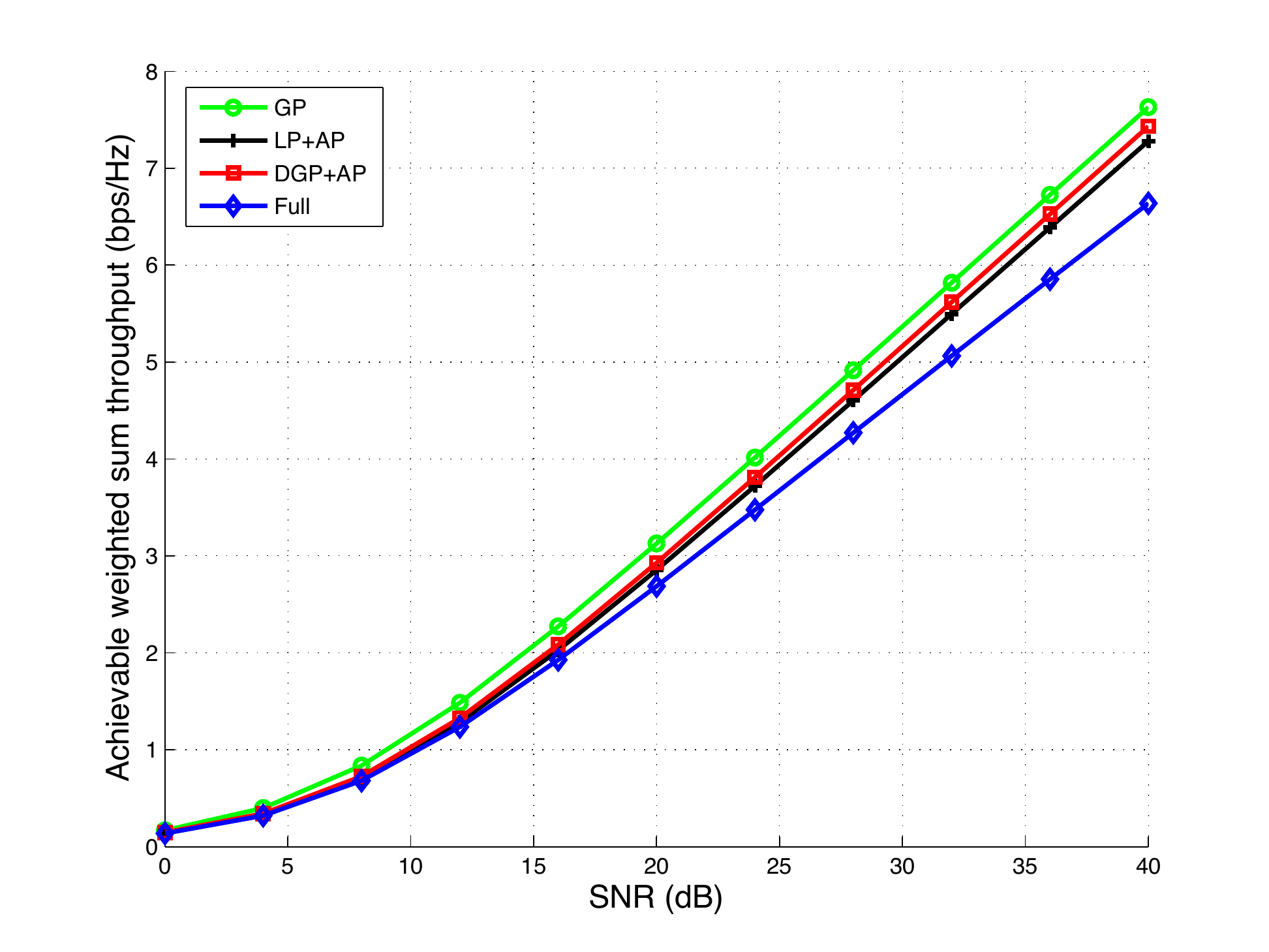}
\includegraphics[width=0.5\columnwidth]{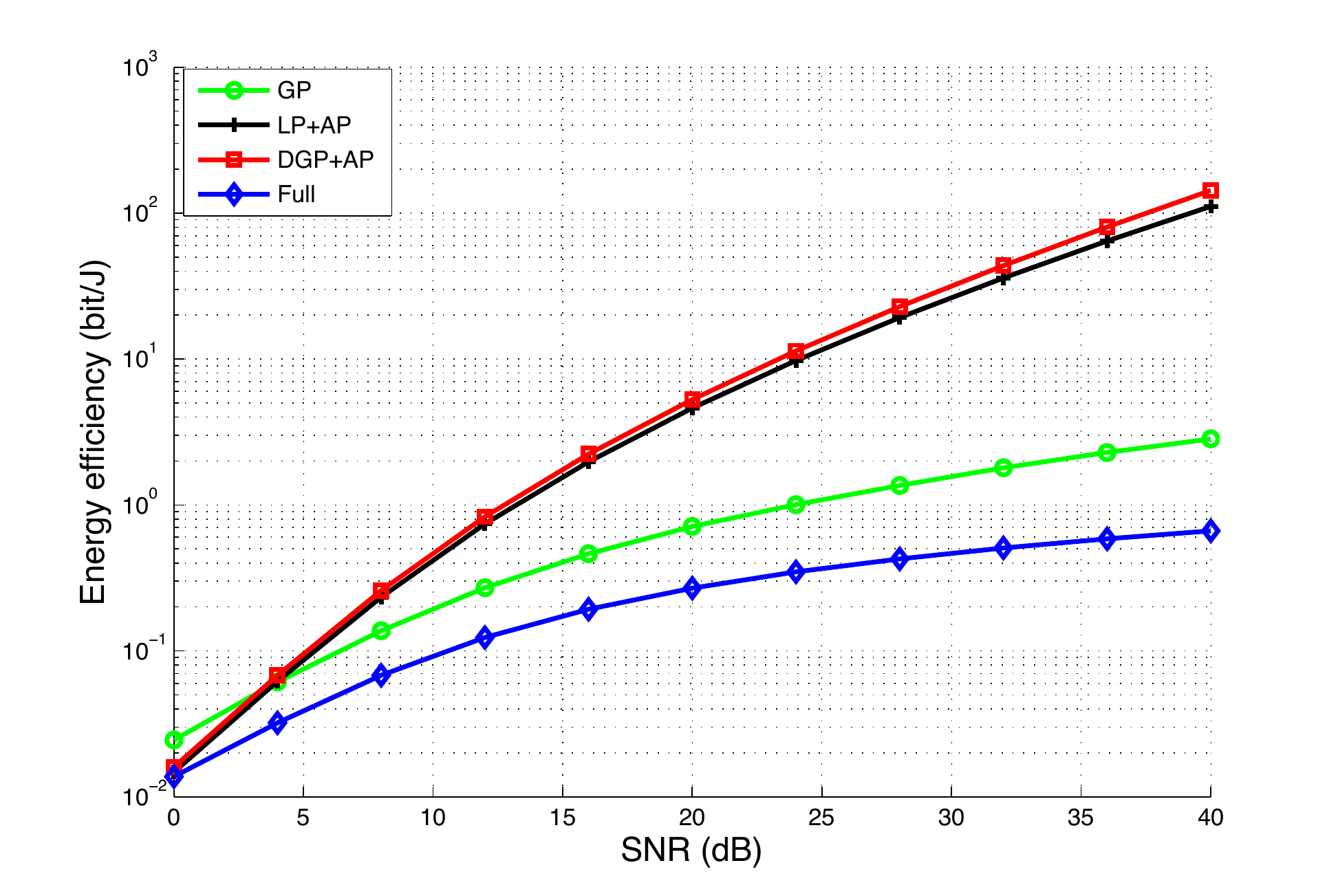}
\caption{Sum throughput and energy efficiency versus SNR.}
\label{fig:vssnr}
\end{figure}

Four schemes are compared: (1) Conventional GP-based power allocation (``GP''), (2) distributed power allocation via DGS-Auction algorithm where the GDoF tuple is found by LP (``LP+AP''), (3) distributed power allocation via DGS-Auction algorithm where the GDoF tuple is found by the distributed implementation of GP (``DGP+AP''), and (4) full power allocation (``full'') for reference. It is clear that the conventional GP-based power control offers the maximal sum throughput, and the assignment-inspired power control with the GDoF tuple obtained by either GP or LP gives relatively good performance. Notably, the energy-efficiency of the assignment-inspired power control nearly 50 times outperforms that by GP, and 100 times over full power allocation.

\section{Conclusion}
\label{sec:conclusion}
The GDoF optimality problem of treating interference as noise for Gaussian interference channels has been re-formulated in a combinatorial 
optimization perspective. Thanks to this new formulation, we cast power allocation into an assignment problem, which can be solved in polynomial time. 
A new expression for the TIN-Achievable GDoF region is provided, which is more compact and useful than what known before since it eliminates
many redundant inequalities.   A relaxed version of the condition in \cite{TIN} on the channel coefficients is given, for which the TIN-Achievable GDoF region is a convex polytope. 
Finally, a new TIN optimality condition is also revealed, by which TIN still achieves the optimal GDoF region for a class of networks different from the one identified in \cite{TIN}. It is also worth noting that our new TIN optimality condition does not violate the conjecture in \cite{TIN} that the GNAJ condition 
is also necessary ``except for a set of channel gain values with measure zero.''

We are also able to translate these insights into practical communication systems (e.g., D2D networks). With the newly found channel strength condition, we employed it as a new decision criterion in a distributed link scheduling mechanism. Together with the globally optimal distributed power allocation algorithms, we proposed a distributed spectrum sharing mechanism (ITLinQ+) for D2D networks. By simulation, we have shown that our ITLinQ+ mechanism achieves significant sum throughput improvement over FlashLinQ and ITLinQ with the same level implementation complexity. Moreover, ITLinQ+ also promises a substantial improvement on energy efficiency.

\section*{Appendix}
\subsection{Preliminaries}
\label{sec:pre}
\subsubsection{Weighted Matching}
In this work we shall make extensive use of weighted matchings \cite{matching} of bipartite graphs. 
We recall here some basic definitions.
Let $\Gc=(\Uc,\Vc,\Ec)$ denote a bipartite graph with left vertices $\Uc$, right vertices $\Vc$ and edges $\Ec \subseteq \Uc \times \Vc$. 
A matching $\Mc \subseteq \Ec$ is a set of edges, any two of which do not share the same vertex.
When weights $w(u,v)$ are associated to the edges $(u,v) \in \Ec$, we denote by $w(\Mc) = \sum_{(u,v) \in \Mc} w(u,v)$ the weight of the matching $\Mc$, and
we let $\Mc^* = \arg_{\Mc} \max w(\Mc)$ denote the maximum weighted matching, i.e., the matching with maximum sum-weight.  
$\Mc^*$ can be characterized as the solution of the integer program: 
\begin{align}
\max \quad & \sum_{(u,v) \in \Ec} w(u,v) x(u,v),  \label{frac-matching1} \\
{\rm s.t.} \quad & \sum_{u \in \Uc :  (u,v) \in \Ec} x(u,v) \le 1,  \label{frac-matching2} \\
               & \sum_{v \in \Vc :  (u,v) \in \Ec} x(u,v) \le 1,  \label{frac-matching3} \\
               & x(u,v) \in \{0,1\}.  \label{frac-matching4}
\end{align}
When equality holds in all constraints,  the resulting solution is called a {\em perfect matching}, i.e., a matching that covers all vertices.
The LP relaxation of (\ref{frac-matching1})-(\ref{frac-matching4}), obtained by  replacing (\ref{frac-matching4}) with
$x(u,v) \in [0,1]$, is called fractional matching \cite{Fractional2011}. For bipartite graphs, the solution of this LP relaxation is integral, i.e., $x \in \{0,1\}$, meaning that, given a fractional matching, there exists a perfect matching such that the sum-weights of two matchings are equal. In other words, there always exists an integral solution to the LP relaxation problem.

A vertex/edge is called matched if it is involved in a matching; otherwise it is a free vertex/edge. A path is alternating if its edges alternate between matched and free edges. The augment operation ${\rm aug(\cdot)}$ is to exchange matched and free edges in an alternating path that starts from and ends to free vertices. For instance, given an alternating path $\Pc=\{(i_0,i_1),(i_1,i_2),(i_2,i_3),\dots,(i_{2n},i_{2n+1})\}$ consists of a matching  $\Mc=\{(i_1,i_2), (i_3,i_4),\dots,(i_{2n-1},i_{2n})\}$ and free edges $\Pc \backslash \Mc$, the augment of $\Pc$ results in a new matching $\Mc'={\rm aug} (\Pc) \defeq \Pc \backslash \Mc$ and free edges $\Mc$.

\subsubsection{Geometric Programming}
Geometric programing is a powerful tool to solve a class of nonlinear optimization problems under a standard form
\begin{subequations}
\begin{align}
\min \quad & f_0(\xv)\\
{\rm s.t.} \quad & f_i(\xv) \le 1, \quad i=1,\dots,m\\
& g_i(\xv)=1, \quad i=1,\dots,p
\end{align}
\end{subequations}
where $\{f_i(\xv),i=0,1,\dots,m\}$ are posynomial functions $f_i(\xv): \mathbb{R}^n \mapsto \mathbb{R}$ in a form of
\begin{align}
f(x_1,\dots,x_n) = \sum_{k=1}^K c_k x_1^{a_{1k}} x_2^{a_{2k}} \cdots x_n^{a_{nk}}
\end{align}
with $c_k \ge 0, \forall k$ and $\{g_i(\xv),i=1,\dots,p\}$ are monomial functions $g_i(\xv): \mathbb{R}^n \mapsto \mathbb{R}$ in the form of
\begin{align}
g(x_1,\dots,x_n) = c x_1^{a_{1}} x_2^{a_{2}} \cdots x_n^{a_{n}}
\end{align}
with $c \ge 0$.

{ 
\subsection{The Kuhn-Munkres Algorithm}
\label{sec:kuhn}
To ease the presentation, we construct a bipartite graph $\Gc=(\Uc,\Vc)$ with weight of edge $(i,j)$ being $\Am_{ij}$ and $\Uc,\Vc$ being transmitter and receiver sets respectively. The Kuhn-Munkres algorithm is to find the maximum weighted matching in this bipartite graph. The input is the weight matrix $\Am$ defined in \eqref{eq:inputmatrix}, and the output is the matching with maximum sum weights and the corresponding left and right labels $\{y_{u_j}, y_{v_j}\}_j$, in which the left labels $\{y_{u_j}\}_j$ achieve the maximum left label equilibrium \cite{Hungarian1955,munkres1957}.

As the initialization, we choose a feasible labeling with $y_{u_i}=\max_j {\Am_{ij}}, \forall~i$, and $y_{v_j}=0, \forall~j$. This labeling is feasible, because $y_{u_i} + y_{v_j} \ge \Am_{ij}$ always holds for any pair of $i \in \Uc$ and $j \in \Vc$. We also construct an equality subgraph, $\Gc_E$, including all the vertices of $\Gc$ but only those edges $(i,j)$ such that  
\begin{align}
y_{u_i} + y_{v_j} = \Am_{ij}, \forall i \in \Uc,j \in \Vc.
\end{align}
It has been proved in \cite{Hungarian1955,munkres1957} that, once $\Gc_E$ has a perfect matching $\Mc$, then this matching $\Mc$ is a maximum weighted matching, and thus the corresponding labels are the final solution to the assignment problem.

The algorithm consists of multiple rounds of iterations. 
In each round, we first check if there exists a perfect matching in $\Gc_E$. Note that for a feasible GDoF tuple, the maximum matching always involves edges $(j,j)$ according to \eqref{scummily}. Thus, the perfect matching consists of edges $\{(j,j), \forall j \in \Kc\}$.
If not, the left and right labels are carefully decreased and increased respectively, and the equality subgraph $\Gc_E$ is updated accordingly. Once $\Gc_E$ contains the perfect matching $\{(j,j), \forall j \in \Kc\}$, the resulting left labels $\{y_{u_j}\}_j$ achieve the maximum left label equilibrium, yielding the global minimal power allocation $r_j = -y_{u_j}$ for all $j$. The details of the Kuhn-Munkres algorithm are given in Algorithm~\ref{alg:Hung}, where $\Nc_L(\Sc)$ is the neighborhood of a set of nodes of $\Sc$ in $\Gc_E$, i.e., $\Nc_L(\Sc)=\{j: (i,j) \in \Gc_E, \forall i \in \Sc, j \in \Vc\}$. An illustrative example is also given as follows.
\begin{algorithm}
\caption{A Centralized Power Allocation Algorithm via the Hungarian Method}
\label{alg:Hung}
\begin{algorithmic}[1]
\Require Matrix $\Am$ with $ij$-th element specified in \eqref{eq:inputmatrix}.
\State Initialization: Set $y_{u_i}=\max_j {\Am_{ij}}, \forall~i$ and $y_{v_j}=0, \forall~j$. Construct $\Gc_E$ according to $\{y_{u_i},y_{v_j}, \forall i,j\}$ and choose an arbitrary matching $\Mc$ contained in $\Gc_E$.
 Let $\Sc=\Tc=\emptyset$.
\If {$\exists~\Mc = \{(j,j), \forall j \in \Kc\}$ in $\Gc_E$} 
\State $r_i = -y_{u_i}, \forall~i$, and \Return
\Else  
\State Pick a free vertex $u \in \Uc$
\State $\Sc \gets \{u\}$, $\Tc \gets \emptyset$.
\EndIf

\If {$\Nc_L(\Sc)=\Tc$}
    $\alpha_L = \min_{i \in \Sc, j \notin \Tc} \{ y_{u_i} + y_{v_j} - \Am_{ij}\}, $
\State   Update $y_{u_k} \gets y_{u_k} - \alpha_L$, if $k \in \Sc$, 
\State   Update $y_{v_k} \gets y_{v_k}+\alpha_L$, if $k \in \Tc$,
\State   Update $\Mc=\{(i,j): y_{u_i} + y_{v_j} = \Am_{ij}\}$
\Else
  \State Pick $v \in \Nc_L(\Sc) \backslash \Tc$
   \If {$v$ is a free vertex} 
\State   Augment the alternating path $u \to v$ that contains the matching $\Mc$
\State Update $\Mc \gets {\rm aug}(\{u \to v\})$ and goto 2
   \Else \If {$v$ is matched to $u'$} 
\State   $\Sc \gets \Sc \cup \{u'\}, \Tc \gets \Tc \cup \{v\}$
\State goto 8
   \EndIf
   \EndIf
\EndIf
\end{algorithmic}
\end{algorithm}


\begin{example}
\label{ex:Kuhn}
\normalfont
The detailed power allocation procedure according to Algorithm \ref{alg:Hung} is presented as follows.

\begin{figure}[htb]
\centering
\includegraphics[width=0.7\columnwidth]{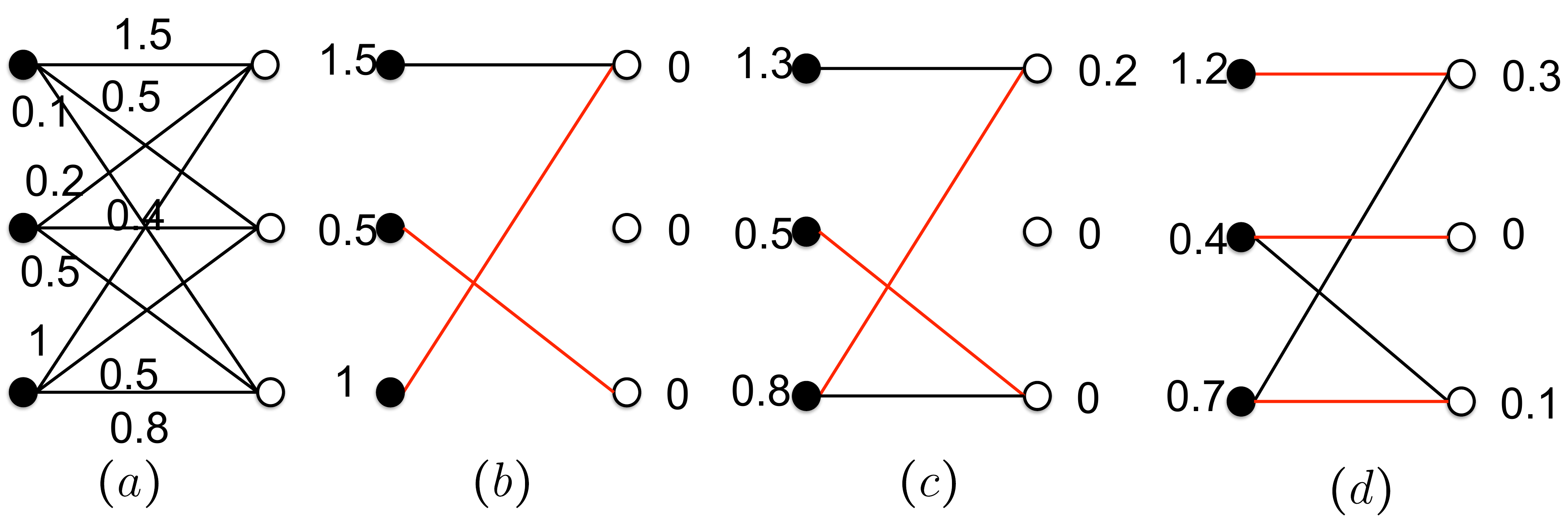}
\caption{ Illustration of the Hungarian method. (a) The weighted bipartite graph with weight of $ij$-th edge being $\Am_{ij}$. 
 (b) The equality subgraph with initiated labeling. (c) The updated equality subgraph with first round label update. (d) The update equality subgraph with second round label update. The edges in red give a matching in equality subgraphs. }
\label{fig:hungarian}
\end{figure}

As an initialization, we assign
\begin{align}
y_{u_1}=1.5, \; y_{u_2}=0.5, \; y_{u_3}=1, \; y_{v_1}=y_{v_2}=y_{v_3}=0
\end{align}
such that we construct the equality subgraph $\Gc_E$ with edges $\{(1,1),(2,3),(3,1)\}$ as in Fig. \ref{fig:hungarian}(b). Note that $\Gc_E$ does not contain a perfect matching. So we choose an arbitrary matching, e.g., $\Mc=\{(2,3),(3,1)\}$, as shown in Fig. \ref{fig:hungarian}(b).

In the first round, we choose a free vertex $u_1$ in $\Gc_E$ and set $\Sc=\{u_1\}$ and $\Tc=\emptyset$.
Because $\Nc_L(\Sc)=\{v_1\} \neq \Tc$, we go to line-12 and pick $v_1 \in \Nc_L(\Sc) \backslash \Tc$.
Note that $v_1$ is matched to $u_3$, and thus we update $\Sc=\{u_1,u_3\}$, $\Tc=\{v_1\}$.
As $\Nc_L(\Sc) = \{v_1\}=\Tc$, we go to line-8 and obtain
\begin{align}
\alpha_L = \min_{u_i \in \{u_1,u_3\}, v_j \in \{v_2,v_3\}} \{y_{u_i} + y_{v_j} - \Am_{ij}\} = 0.2
\end{align}
As such, we have
\begin{align}
y_{u_1}=1.3, \; y_{u_2}=0.5, \; y_{u_3}=0.8, \; y_{v_1}=0.2, \; y_{v_2}=0, \; y_{v_3}=0.
\end{align}
and go to line-2.

In the second round, we update the equality subgraph $\Gc_E$ with edges $\{(1,1),(2,3),(3,1),(3,3)\}$ and there still does not contain a perfect matching, as shown in Fig. \ref{fig:hungarian}(c). Thus, we choose $\Mc=\{(2,3),(3,1)\}$ as a matching. Still, we pick the free vertex $u_1$, and set $\Sc=\{u_1\}$ and $\Tc=\emptyset$. Now, as $\Nc_L(\Sc)=\{v_1,v_3\} \neq \Tc$, we pick $v_1 \in \Nc_L(\Sc) \backslash \Tc$. Because $v_1$ is matched to $u_3$, we update $\Sc=\{u_1,u_3\}$, $\Tc=\{v_1\}$. At this point, $\Nc_L(\Sc) = \{v_1,v_3\} \ne \Tc$ again. Thus, we pick $v_3 \in \Nc_L(\Sc) \backslash \Tc$. Due to $v_3$ is matched to $u_2$, we update $\Sc=\{u_1,u_2,u_3\}$ and $\Tc=\{v_1,v_3\}$. Here $\Nc_L(\Sc) = \{v_1,v_3\}=\Tc$, we go to line-8 and have
\begin{align}
\alpha_L = \min_{u_i \in \{u_1,u_2,u_3\}, v_j \in \{v_2\}} \{y_{u_i} + y_{v_j} - \Am_{ij}\} = 0.1
\end{align}
As such, the labels are updated as
\begin{align}
y_{u_1}=1.2, \; y_{u_2}=0.4, \; y_{u_3}=0.7, \; y_{v_1}=0.3, \; y_{v_2}=0, \; y_{v_3}=0.1.
\end{align}
and then we go to line-2.

Till now, in the updated equality subgraph shown in Fig. \ref{fig:hungarian}(d), we have a perfect matching $\{(1,1),(2,2),(3,3)\}$. Thus, the algorithm returns with
\begin{align}
r_1=-1.2, \; r_2=-0.4, \; r_3=-0.7.
\end{align}
\end{example}
}

\subsection{Proof of Theorem~\ref{theorem:tina-represent}}
\label{proof:tina-represent}
For the sake of this proof, we denote by $\Pc_\Sc$ the region defined by (\ref{TINA-Jafar}), and by $\Pc^{\rm TINA}_\Sc$ the region defined
by (\ref{ziofa}). Our goal is to show that $\Pc_\Sc = \Pc^{\rm TINA}_\Sc$ for any $\Sc \subseteq \Kc$. 

\subsubsection*{\underline{$\Pc_{\Sc} \subseteq \Pc^{\rm TINA}_{\Sc}$}}
To prove this, we show that for any inequality presented in $\Pc^{\rm TINA}_{\Sc}$, we can always find the same one in $\Pc_{\Sc}$.

Given a subnetwork $\Gc[\Sc]$, the matching with the maximum weight is a degree-1 subgraph. \footnote{If direct links are in the matching, we can eliminate them from $\Sc$, which does not affect our proof.} Connecting the direct links will lead to single or multiple disjoint cycles, as all the nodes has degree-2. 

For the single-cycle case, this cycle corresponds to a sum GDoF constraint in $\Pc_{\Sc}$.
For the multiple-cycle case, each cycle corresponds to a sum GDoF constraint in $\Pc_{\Sc}$ of the users involved in this cycle. Thus, the sum GDoF constraint with the maximum weighted matching in $\Pc^{\rm TINA}_{\Sc}$ corresponds to the combination of these sum GDoF constraints in $\Pc_{\Sc}$.

As such, $\Pc_{\Sc}$ contains or implies all the constraints in $\Pc^{\rm TINA}_{\Sc}$. As $\Pc_{\Sc}$ has more constraints, it follows that $\Pc_{\Sc} \subseteq \Pc^{\rm TINA}_{\Sc}$.

\subsubsection*{\underline{$\Pc^{\rm TINA}_{\Sc} \subseteq \Pc_{\Sc}$}}
To prove this, we show that, for any subset of users $\Sc$, the TINA GDoF region confined by $\Pc_{\Sc}$ is no larger than that by $\Pc^{\rm TINA}_{\Sc}$. It is clear that $\Pc_{\Sc}$ is determined by individual GDoF and sum GDoF constraints of any subset of users in $\Sc$. The individual GDoF constraints of two regions are identical. Thus, our focus will be on the sum GDoF constraints for users in $\Sc$ with $\abs{\Sc} \ge 2$.

For the user set $\Sc$, the sum GDoF constraints in $\Pc_{\Sc}$ only come from (1) the sum GDoF constraints with all possible permutations of $\Sc$, and (2) the combination of a number of individual and/or sum GDoF constraints of subsets of $\Sc$.
For the first case, the sum GDoF constraint in $\Pc_{\Sc}$ is dominated by the maximum weight of any possible matchings (associated with cyclic sequences). 
For the second case, suppose the combination involves a number of subnetworks, where the subnetworks may have intersections. This combination of constraints involves every user with equal times (say $b$ times), otherwise, the combination will not lead to a sum GDoF constraint, because it is a weighted sum GDoF constraint and can be implied by the combination of other sum GDoF constraints. 
Each sum GDoF constraint for a subnetwork involves a cyclic sequence and hence forms a matching. Thus, the combination of sum GDoF constraints corresponds to a fractional perfect matching by assigning $x(u,v)$ in \eqref{frac-matching1}-\eqref{frac-matching4} with $\frac{1}{b}$. In bipartite graphs, the weight of any fractional perfect matching equals the weight of a perfect matching \cite{matching,Fractional2011}.

Thus, neither the weight of any matching nor of any fractional matching is greater than the maximum weighted matching, such that the sum GDoF constraints in $\Pc^{\rm TINA}_{\Sc}$ will be more restrictive than those or any combinations in $\Pc_{\Sc}$, i.e., $\Pc^{\rm TINA}_{\Sc} \subseteq \Pc_{\Sc}$. This completes the proof.

\subsection{Proof of Theorem~\ref{lemma:sim_region}}
\label{proof:sim_region}
In what follows, we prove that under condition (\ref{C1}), 
$\Pc^{\rm TINA}_{\Sc}$ is monotonically increasing. 
Hence, from (\ref{union-tina}) this immediately implies that $\Rc^{\rm TINA} = \Pc^{\rm TINA}_{\Kc}$ which, 
by inspection, is a convex polytope. 

Let us start with $\abs{\Sc}=2$. Suppose without loss of generality $\Sc=\{k,j\}$.
Due to the condition (\ref{C1}), $\min\{\alpha_{kk}, \alpha_{jj}\} \ge \alpha_{kj} + \alpha_{jk}$, then it is easy to verify that $\Pc^{\rm TINA}_k \subseteq \Pc^{\rm TINA}_{\{k,j\}}$ and $\Pc^{\rm TINA}_j \subseteq \Pc^{\rm TINA}_{\{k,j\}}$. 

Then, we prove the general cases with the following lemma.
\begin{lemma} \label{lemma:diff}
Given a subgraph $\Gc[\Sc]$ with weights $\{\alpha'_{ij},i,j \in \Sc\}$, the difference of maximum weighted matching with and without the user $k$ is bounded by
\begin{align*}
w (\Mc^*_\Sc) - w(\Mc^*_{\Sc \backslash \{k\}}) \le \max_{i,j \in \Sc, i,j \ne k} \{\alpha_{ik} + \alpha_{kj} - \alpha'_{ij}\}
\end{align*}
\end{lemma}
\begin{proof}
Suppose without loss of generality that the maximum weighted matching of $\Gc[\Sc]$ $(k \in \Sc)$ includes links $(i,k)$ and $(k,j)$ with weights $\alpha'_{ik}$ and $\alpha'_{kj}$ respectively and $i,j \ne k$. Note that whether $i=j$ or not does not affect our proof. After removing user $k$ and edges $(i,k)$, $(k,j)$ from the matching, and adding the link $(i,j)$ with weight $\alpha'_{ij}$, we have a matching for $\Sc \backslash \{k\}$. Thus, for all $\{(i,k),(k,j)\} \in \Mc^*_{\Sc}$, we have
\begin{align}
\MoveEqLeft w (\Mc^*_{\Sc}) - w (\Mc^*_{\Sc \backslash \{k\}}) \nn \\ &\le \max_{\{(i,k),(k,j)\} \in \Mc^*_{\Sc}, i,j \in \Sc, i,j \ne k}  \{ \alpha'_{ik} + \alpha'_{kj} - \alpha'_{ij} \} \\
&\le \max_{i,j \in \Sc, i,j \ne k}  \{ \alpha'_{ik} + \alpha'_{kj} - \alpha'_{ij} \}\\
&=\max_{i,j \in \Sc, i,j \ne k}  \{ \alpha_{ik} + \alpha_{kj} - \alpha'_{ij} \}.
\end{align} 
\end{proof}

Together with the condition (\ref{C1}), we have
\begin{align}
\alpha_{kk} \ge w (\Mc^*_\Sc) - w(\Mc^*_{\Sc \backslash \{k\}})
\end{align}
for any user $k \in \Sc$. Thus, by comparing the sum GDoF constraints without and with user $k$
\begin{align}
\sum_{j \in \Sc \backslash \{k\}} d_j &\le \sum_{j \in \Sc \backslash \{k\}} \alpha_{jj} - w (\Mc^*_{\Sc \backslash \{k\}})\\
\sum_{j \in \Sc} d_j &\le \sum_{j \in \Sc} \alpha_{jj} - w (\Mc^*_\Sc)\\
&=\sum_{j \in \Sc \backslash \{k\}} \alpha_{jj} + \alpha_{kk} - w (\Mc^*_\Sc),
\end{align}
it is readily verified that as long as (\ref{C1}) is satisfied, the sum GDoF constraint with user $k$ is not implied by the sum GDoF constraint without user $k$ and the individual GDoF constraint $d_k \le \alpha_{kk}$. In other words, with user $k$, the GDoF region is not decreasing. It follows immediately that $\Pc^{\rm TINA}_{\Sc \backslash \{k\}} \subseteq \Pc^{\rm TINA}_{\Sc}$ $(\forall k\in \Sc)$. More generally, if $\Sc_1 \subseteq \Sc_2$, then $\Pc^{\rm TINA}_{\Sc_1} \subseteq \Pc^{\rm TINA}_{\Sc_2}$. This completes the proof.

\subsection{Proof of Theorem~\ref{theorem:tin}}
\label{proof:tin}
Due to the fact that $\Rc^* \supseteq \Rc^{\rm TINA}$ and that, 
under condition (\ref{C1}), $\Rc^{\rm TINA} = \Pc^{\rm TINA}_{\Kc}$, achievability trivially follows.

For the converse, we follow the cyclic outer bounds first revealed in \cite[Theorem 2]{cyclicGIC} and later used to prove the 
optimality of TIN condition in \cite[Theorem 3]{TIN}. 

Thus, for the $K$-user Gaussian interference channel in the weak interference regime, the GDoF region under the condition \eqref{C1} is included in the set of GDoF tuples $(d_1,d_2,\dots,d_K)$ such that
\begin{align}
d_j &\le \alpha_{jj}, \quad \forall j \in \Kc\\ 
 \label{papa}
\sum_{j=0}^{m-1} d_{i_j} &\le \min \{f_\pi, g_{\pi,0}, \dots, g_{\pi,{m-1}}\},
\end{align} 
for any ordered subset $\pi = (i_0,i_1,\dots,i_{m-1}) \subset \Kc^m$, where we define
\begin{align}
 f_{\pi} &\defeq \sum_{j=0}^{m-1} \max\{0, \alpha_{i_j i_{j+1}}, \alpha_{i_j i_j} - \alpha_{i_{j-1} i_j}\} \\
g_{\pi,k}
&\defeq  \sum_{j=0}^{m-1} (\alpha_{i_j i_j} - \alpha_{i_{j-1} i_j}) + \alpha_{i_{k-1} i_k}, \;\; k = 0,\ldots, m-1, \label{outbound2}
\end{align}
and where the index subscript arithmetic is modulo $m$.

When $m=2$, then condition (\ref{C1}) is equivalent to the GNAJ condition and the bound is known to be tight. 
When $m > 2$, let us first consider the bound formed by the ``g'' terms in (\ref{papa}).
Notice that the left-hand side of (\ref{papa}) depends only on the indices in $\pi$ but not on its order.
Hence, letting $\Sc$ denote a given unordered subset of size $m$ of $\Kc$ and using the short-cut notation $\pi \in \pi(\Sc)$ to indicate
the ordered sets formed with the elements of $\Sc$, i.e., the permutations of $\Sc$, we can write
\begin{align}
\MoveEqLeft \min_{\pi \in \pi(\Sc)} \min_{k = 0,\ldots, m-1} \{ g_{\pi,k} \} \nn \\
&=  \min_{\pi \in \pi(\Sc)} \left \{ \sum_{j=0}^{m-1} (\alpha_{i_j i_j} - \alpha_{i_{j-1} i_j}) +  
\min_{k = 0,\ldots, m-1} \{\alpha_{i_{k-1} i_k}\} \right \} \nn \\
&= \sum_{j \in \Sc} \alpha_{jj} - \max_{\pi \in \pi(\Sc)} \left \{ \sum_{j=0}^{m-1}  \alpha_{i_{j-1} i_j} - \min_{k = 0,\ldots, m-1} \{\alpha_{i_{k-1} i_k}\} \right \}  \nn\\
&=  \sum_{j\in \Sc} \alpha_{j j} - w(\Mc^*_{\Sc}) \label{eq:zero}
 \end{align}
where \eqref{eq:zero} is due to the condition (\ref{C2}). 
If the maximum weighted matching involves multiple cycles, 
then the sum GDoF outer bound can be the combination of multiple sum GDoF constraints associated 
with the corresponding cyclic sequences. 
Thus, \eqref{eq:zero} still holds, because condition (\ref{C2}) 
holds for any subset of $\Sc \subseteq \Kc$. Due to the fact that
 \begin{align*}
 \sum_{j=0}^{m-1} (\alpha_{i_j i_j} - \alpha_{i_{j-1} i_j}) \le \sum_{j=0}^{m-1} \max\{0, \alpha_{i_j i_{j+1}}, \alpha_{i_j i_j} - \alpha_{i_{j-1} i_j}\} 
 \end{align*}
we have that
\begin{align}
\sum_{j \in \Sc} d_j &\le \min_{\pi \in \pi(\Sc)} \min \{f_\pi, g_{\pi,0}, \dots, g_{\pi,m-1}\}\\
&=  \sum_{j \in \Sc} \alpha_{j j} - w(\Mc^*_{\Sc})
\end{align} 
which coincides with $\Pc^{\rm TINA}_\Sc$ for every $\Sc \subseteq \Kc$. Under the condition (\ref{C1}), 
the TINA is the largest polyhedral region, so the converse bound is tight.
%


\end{document}